\documentclass[12pt]{amsart}
\usepackage{epsf,amscd,amsmath,amssymb,amsfonts,mathtools}

\usepackage{graphicx}
\usepackage{epstopdf}
\usepackage{fullpage}
\usepackage[utf8]{inputenc}
\usepackage{fullpage}
\usepackage{latexsym}
\usepackage{tikz}
\usepackage{color}

\theoremstyle{plain}
\newtheorem{thm}{Theorem}[section]
\newtheorem*{thm*}{\bf Theorem }

\newtheorem{prop}[thm]{Proposition}
\newtheorem*{prop*}{\bf Proposition}

\theoremstyle{definition}

\theoremstyle{remark}
\newtheorem{rem}[thm]{Remark}
\newtheorem{rem*}[thm]{Remark}

\numberwithin{equation}{section}

\newcommand{\be}{\begin{equation}}
\newcommand{\ee}{\end{equation}}
\newcommand{\bea}{\begin{eqnarray}}
\newcommand{\eea}{\end{eqnarray}}
\newcommand{\beas}{\begin{eqnarray*}}
\newcommand{\eeas}{\end{eqnarray*}}

\theoremstyle{plain}
\theoremstyle{definition}

\numberwithin{thm}{section}
\numberwithin{equation}{section}

\newcommand{\V}{{\mathbb V}}

\newcommand{\Z}{{\mathbb Z}}

\def\One{\mathbb{I}}

\def\bthree{{\;\raisebox{-20mm}{\includegraphics[height=60mm]{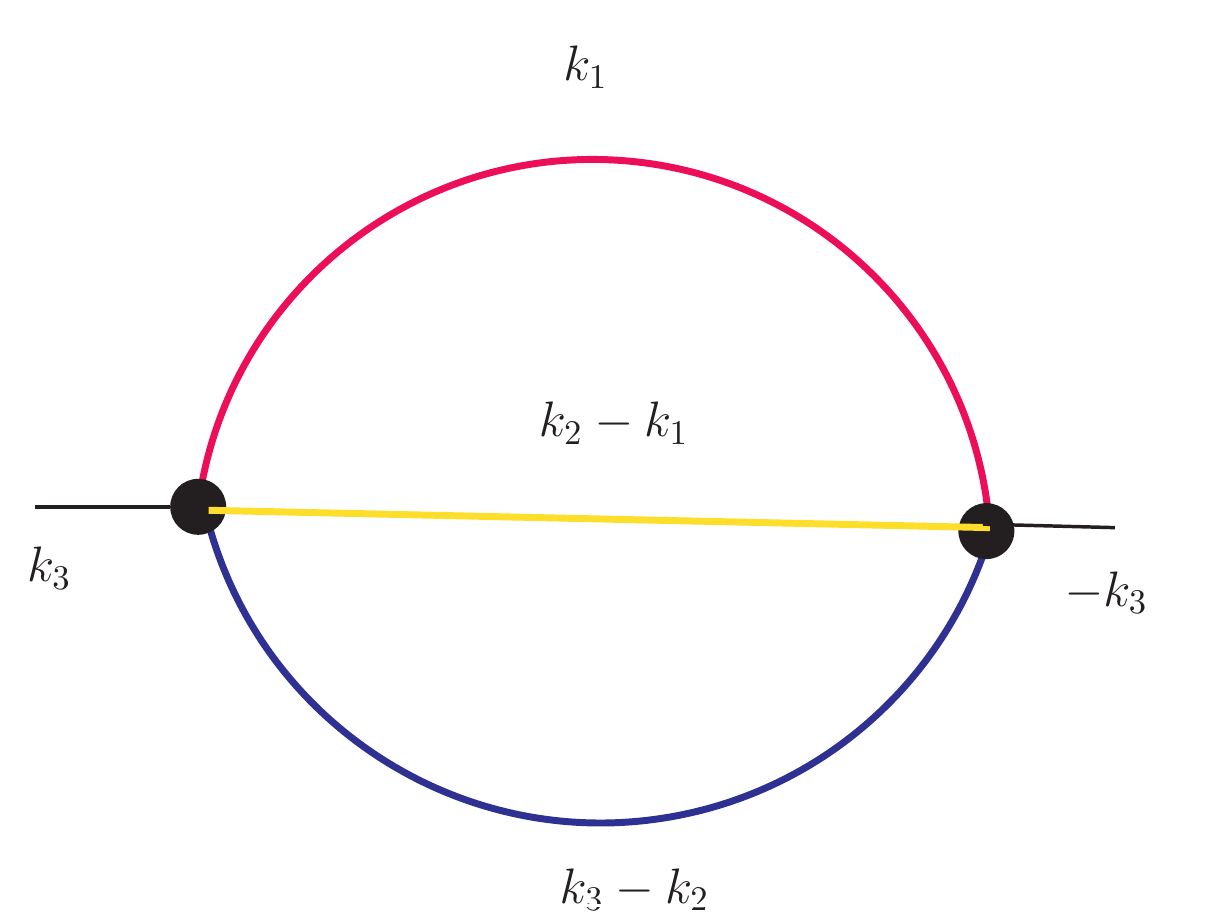}}\;}}

\usepackage{float}
\floatstyle{boxed} 
\restylefloat{figure}
\usepackage{hyperref}

\title{Bananas: multi-edge graphs and their Feynman integrals}
\author{Dirk Kreimer}
\address{Humboldt U.\ Berlin, Unter den Linden 6, 10099 Berlin, Germany}
\begin{document}
\begin{abstract}We consider multi-edge or banana graphs $b_n$ on $n$ internal edges $e_i$ with different masses $m_i$. We focus on the cut banana graphs $\Im(\Phi_R(b_n))$ from which the full result $\Phi_R(b_n)$ can be derived through dispersion. 
We give a recursive definition of $\Im(\Phi_R(b_n))$ through iterated integrals.
We discuss the structure of this iterated integral in detail. A discussion of accompanying differential equations, of monodromy and of a basis of master integrals is included. 
\end{abstract}

\maketitle
\tableofcontents
\section*{Acknowledgments}
This is work originating from discussions with Karen Vogtmann and Marko Berghoff which are gratefully acknowledged. I thank Spencer Bloch,  David Broadhurst and Bob Delbourgo for friendship and for sharing insights into the mathematics and physics of quantum field theory over the years. And David for pointing out some older literature. Enjoyable discussions with Ralph Kaufmann on possible similarities of the structure  of phase-space integrals and his use of singularity theory  in applied quantum field theory \cite{Ralph} were a welcome stimulus to write these results.  
\section{Introduction}
We define a banana graph $b_n$ by two vertices $v_1,v_2$ connected by $n$ edges
forming a multi-edge.\footnote{Often $b_2$ is called a bubble, $b_3$ a sunset and $b_4$ a banana graph. We call all $b_n$, $2\leq n< \infty$ banana graphs.}
Furthermore, $v_1,v_2$ are both $n+1$ valent vertices so that $b_n$ has an external edge at each vertex.
\begin{figure}[H]
\includegraphics[width=14cm]{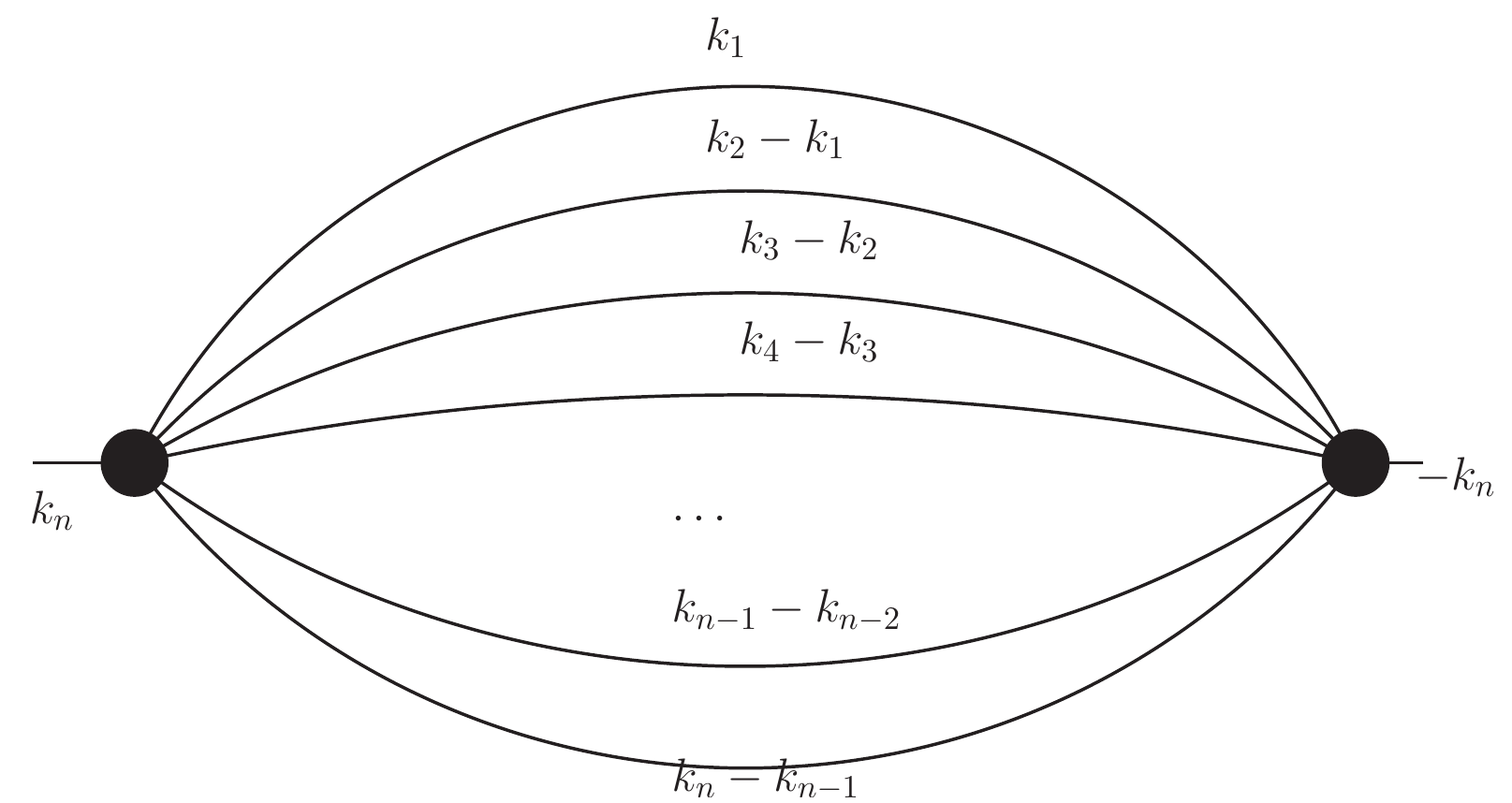}.
\caption{Banana graphs $b_n$ on $|b_n|=(n-1)$ loops. We indicate momenta at internal edges $e_1,\ldots e_n$ labeling from top to bottom. We assign mass square $m_i^2$ to edge $e_i$.
A positive infinitesimal imaginary part is understood in each popagator. Both vertices have an external edge with incoming momenta $k_n$ and $-k_n$. Note that edges $e_1, \ldots,e_j$, $n>j\geq 2$ constitute a banana graph $b_j$ with external momentum $k_{j}$ flowing through. It is a $(j-1)$-loop subgraph of $b_n$. In particular we have a sequence $b_2\subset b_3\subset\cdots \subset b_n$ of graphs which gives rise to an iterated integral.}
\label{bananas}
\end{figure}
\subsection{General considerations}
We study associated banana integrals $\Phi_R^D(b_n)$. The case $n=3$ has been intensively studied and initiated a detailed analysis of elliptic integrals in Feynman amplitudes, see for example  \cite{Veltman,BroedeletalEll,Broedelequalmass,Remetal,RemSchouten,Weinzierl,
Vanhove,Bloch,DavDel,Zay,allBanana}. 
Evaluation at masses $m_i^2\in\{0,1\}\ni k_n^2$ was recognized to provide a rich arena for an analysis of periods in Feynman diagrams \cite{BroadKloster}
including the apperance of elliptic trilogarithms at sixth root of unity in the evaluation of $b_4$ \cite{Bloch}.

Let us pause and put the problem into context.
\subsubsection{Recursion and splitting in phase space integrals}\label{history}
The imaginary part $\Im\left(\Phi_R^D(b_n)\right)$ of $\Phi_R^D(b_n)$
has been a subject of interest for almost seventy years at least \cite{Kerwas,Block,Prem}. This imaginary part has the interpretation of a phase space integral. Our attempt below to express it recursively by an iterated integral can be traced back to this early work. In fact, computing
$\Im\left(\Phi_R^D(b_n)\right)$ by identifying an imaginary part $\Im\left(\Phi_R^D(b_{n-1})\right)$ as a sub-integral amounts to a split in the phasespace integral and this recurses over $n$. 
\subsubsection{Banana integrals and monodromy} 
In the more recent literature the graphs $b_n$ were studied in an attempt to interpretate the monodromies of the associated functions depending on momenta and masses
$\Phi_R^D(b_n)(s,s_0,\{m_i^2\})$ as a generalization of the situation familiar from the study of polylogarithms. This r\^ole of elliptic functions was prominent already in the historical work cited in Sec.(\ref{history}) above and continued to give insights into the structure of phasespace systematically \cite{RemSchouten,DavDel}.
Recently the aim shifted to explore it in the spirit
of modern mathematics. This brought concepts developed in algebraic geometry -motives, Hodge theory, coactions, symbols  and such- to the forefront \cite{BlochVanhove,Vanhove,Bloch,allBanana,Brown,BEK,Broedel}. 
For us the focus is less on elliptic integrals and elliptic polylogarithms prominent in recent work. Rather we focus on the recursive structure of
$\Im\left(\Phi_R^D(b_n)\right)$ as it has a lot to offer still for mathematical analysis.
\subsection{Iterated integral structure for $b_n$}
Our task is to find iterated integral representations for $\Im\left(\Phi_R^D(b_n)\right)$ which give insight into their structure for all $n$. We will use 
$\Im\left(\Phi_R^D(b_2)\right)$ as a seed for the iteration.
$\Im\left(\Phi_R^D(b_3)\right)$ which has $\Im\left(\Phi_R^D(b_2)\right)$
as a subintegral then gives a complete elliptic integral as expected, see Sec.(\ref{bthreeelliptic}).  
Already the computation of $b_4$ indicates more subtle functions to appear as Sec.(\ref{bfournonell}) and Eq.(\ref{bfourythree}) demonstrate. Neverheless it turns out that such functions are very nicely structured
as we explore in Sec.(\ref{bn}).

We want to understand the function $\Phi_R^D(b_n)$ obtained from applying renormalized Feynman rules $\Phi_R^D$ in $D$ dimensions 
\[
\Phi_R^D(b_n)=S_R^\Phi\star\Phi^D(b_n)(s,s_0),
\]
to the graph $b_n$.

We will study in particular the imaginary part $\Im\left(\Phi_R^D(b_n)\right)$ having in mind that $\Phi_R^D(b_n)$ can be obtained from $\Im\left(\Phi_R^D(b_n)\right)$ by a dispersion
integral.

We will mostly work with a kinematic renormalization scheme in which tadpole integrals evaluate to zero. This is particularly well-suited for the use of dispersion. Indeed $\Im\left(\Phi_R^D(b_n)\right)$ is free of short-distance singularities as the $n$ constraints putting $n$ internal propagators on-shell 
fix all non-compact integrations. 

This reduces renormalization of $b_n$ to a mere use of sufficiently subtracted 
dispersion integrals. Correspondingly in kinematic renormalization we can work in a Hopf algebra $H_R=H/I_{\mathrm{tad}}$ of renormalization which divides by the ideal $I_\mathrm{tad}$ spanned by tadpole integrals rendering the graphs
$b_n$ primitive:
\[
\Delta_{H_R}(b_n)=b_n\otimes \One+\One\otimes b_n.
\]
Therefore 
\[
S_R^{\Phi^D}\star\Phi^D(b_n)=\Phi^D(b_n)(s)-\mathit{T}^{(j)}\Phi^D(b_n)(s,s_0).
\]
$\Phi^D$ are the unrenormalized Feynman rules in dimensional regularization
and $\mathit{T}^{(j)}$ is a suitable Taylor operator.

Nevertheless there is no necessity to regulate Feynman integrals in our approach as we can subtract on the level of integrands. Indeed $\mathit{T}^{(j)}$ can be chosen to subtract in the integrand. We implement it below in Eq.(\ref{dispTaylor}) using the dispersion integral. Our conventions for Feynman rules are in App.(\ref{AppFeyn}).

Our interest lies in a compact formula for
\[
\Im\left(\Phi_R^D(b_n)\right)(s,\{m_i^2\})=\int_{\mathbb{M}_n}I_{\mathrm{cut}}(b_n),
\]
with $I_{\mathrm{cut}}(b_n)$ given in Eq.(\ref{icut}).
We will succeed by giving it as an iterated integral in Eq.(\ref{itInt}) below which is part of  Thm.(\ref{monodromyThm}).

Results for $\Phi_R^D(b_n)(s,s_0,\{m_i^2\})$ then follow by (subtracted at $s_0$) dispersion which implements $\mathit{T}^{(\frac{D}{2}-1)(n-1)}$: 
\be\label{dispTaylor}
\Phi_R^D(b_n)(s,s_0,\{m_i^2\})=\frac{(s-s_0)^{(\frac{D}{2}-1)(n-1)}}{\pi}\int_{\left(\sum_{j=1}^n m_j\right)^2}^\infty \frac{\int_{\mathbb{M}_n}I_{\mathrm{cut}}(b_n)(x)}{(x-s)(x-s_0)^{(\frac{D}{2}-1)(n-1)}}dx.
\ee
Note that in the Taylor expansion of
$\Phi_R^D(b_n)(s,s_0,\{m_i^2\})$ around $s=s_0$, the first 
$(\frac{D}{2}-1)(n-1)$ coefficients vanish. These are our kinematic renormalization conditions.

For example 
$\Phi_R^4(b_2)(s_0,s_0)=0$. On the other hand $\Phi_R^2(b_2)(s,s_0)=\Phi_R^2(b_2)(s)$ as it does not need subtraction at $s_0$ as it is ultraviolet convergent.
So $s_0$ disappears from its definition and the dispersion integral is unsubtracted as $(\frac{D}{2}-1)(n-1)=0$ and for $D=6$, $\Phi_R^6(b_2)(s_0,s_0)=0
=\partial_s\Phi_R^6(b_2)(s,s_0)_{|s=s_0}$.   
\subsection{Normal and pseudo-thresholds for $b_n$}
To understand possible choices for $s_0$, define a set $\mathbf{thresh}$ of $2^{n-1}$ real numbers by
\[
\mathbf{thresh}=\{ (\pm m_1\pm\cdots\pm m_n)^2 \},
\]
and set
\[
s_{\mathbf{min}}:=\min\{x\in \mathbf{thresh}\}.
\]
Note that the maximum is achieved by $s_{\mathbf{normal}}:=\left(\sum_{j=1}^n m_j\right)^2$.
Our requirement for $s_0$ is 
\be\label{szero}
s_0\lneq s_{\mathbf{min}}.
\ee
This ensures that the renormalization at $s_0$ does not produce contributions to the imaginary part of the renormalized $\Phi_R^D(b_n)(s,s_0)$ as $\Im(\Phi^D(b_n)(s_0))=0$. 

We call $s_{\mathbf{normal}}$ normal threshold and the $2^{-1}-1$ other elements 
of $\mathbf{thresh}$ preudo-thresholds.

Also we call $m_{\mathbf{normal}}^n:=\sum_{j=1}^n m_j$ the normal mass of $b_n$
and any of the other $2^{n-1}-1$ numbers $|\pm m_1\cdots\pm m_n|$ a pseudo-mass of $b_n$. For any ordering $o$ of the edges of $b_n$ we get a flag $b_2\subset\cdots b_{n-1}\subset b_n$ such that
\[
m_{\mathbf{normal}}^{j+1}=m_{\mathbf{normal}}^j+m_{j+1},\,j\leq n-1.
\]
On the other hand, for any chosen fixed pseudo-mass there exists at least one ordering $o$ of edges of $b_n$ for which the pseudo-mass is $m_1-m_2\pm\cdots$.
\begin{rem}
By the Coleman-Norton theorem \cite{Coleman-N} (or by an analysis of the second Symanzik polynomial $\varphi(b_n)$, see Eq.(\ref{varphibn}) in App.(\ref{apppseudo})) the physical threshold of $b_n$ is 
when the energy $\sqrt{s}$ of the incoming momenta $k_n=(k_{n;0},\vec{0})^T$ equals 
the normal mass
\[
\sqrt{s}=m_{\mathbf{normal}}^n.
\] 
The imaginary part $\Im\left(\Phi_R^D(b_n)\right)$ is then given by the monodromy 
associated to that threshold and is supported at $s\geq m_{\mathbf{normal}}^n$.

In this paper we are mainly interested in the principal sheet monodromy of $b_n$ and hence 
in the monodromy at $\sqrt{s}=m_{\mathbf{normal}}^n$ which gives $\Im(\Phi_R^D(b_n))$. Pseudo-masses are needed to understand monodromy from pseudo-thresholds off the principal sheet.

They can always be expressed as iterated integrals starting possibly from a pseudo-threshold of $\Phi_R^D(b_2)$. Such non-principal sheet monodromies need to be studied to understand the mixed Hodge theory of $\Phi_R^D(b_n)$ as a multi-valued function in future work. See \cite{DirkEll} for some preliminary considerations. 

In preparation to such future work we note that iterated integral representations can also be obtained for pseudo-thresholds in quite the same manner as in Eq.(\ref{itInt}) by changing signs of masses (not mass-squares) 
in  Eq.(\ref{up}) as given in Eq.(\ref{pseudoup})
and correspondingly in the boundaries of the dispersion integral.
This dispersion will then reconstruct variations on non-principal sheets.
We collect these integral representations in App.(\ref{apppseudo}).
\end{rem}\hfill $|$
\section{Banana integrals $\Im\left(\Phi_R^D(b_n)\right)$}
\label{sec:2}
\vspace{1mm}\noindent
\subsection{Computing $b_2$}
We start with the 2-edge banana $b_2$, a bubble on two edges with two different internal masses $m_1,m_2$, indicated by two different colors in Fig.(\ref{btwo}).
\begin{figure}[h]
\includegraphics[width=12cm]{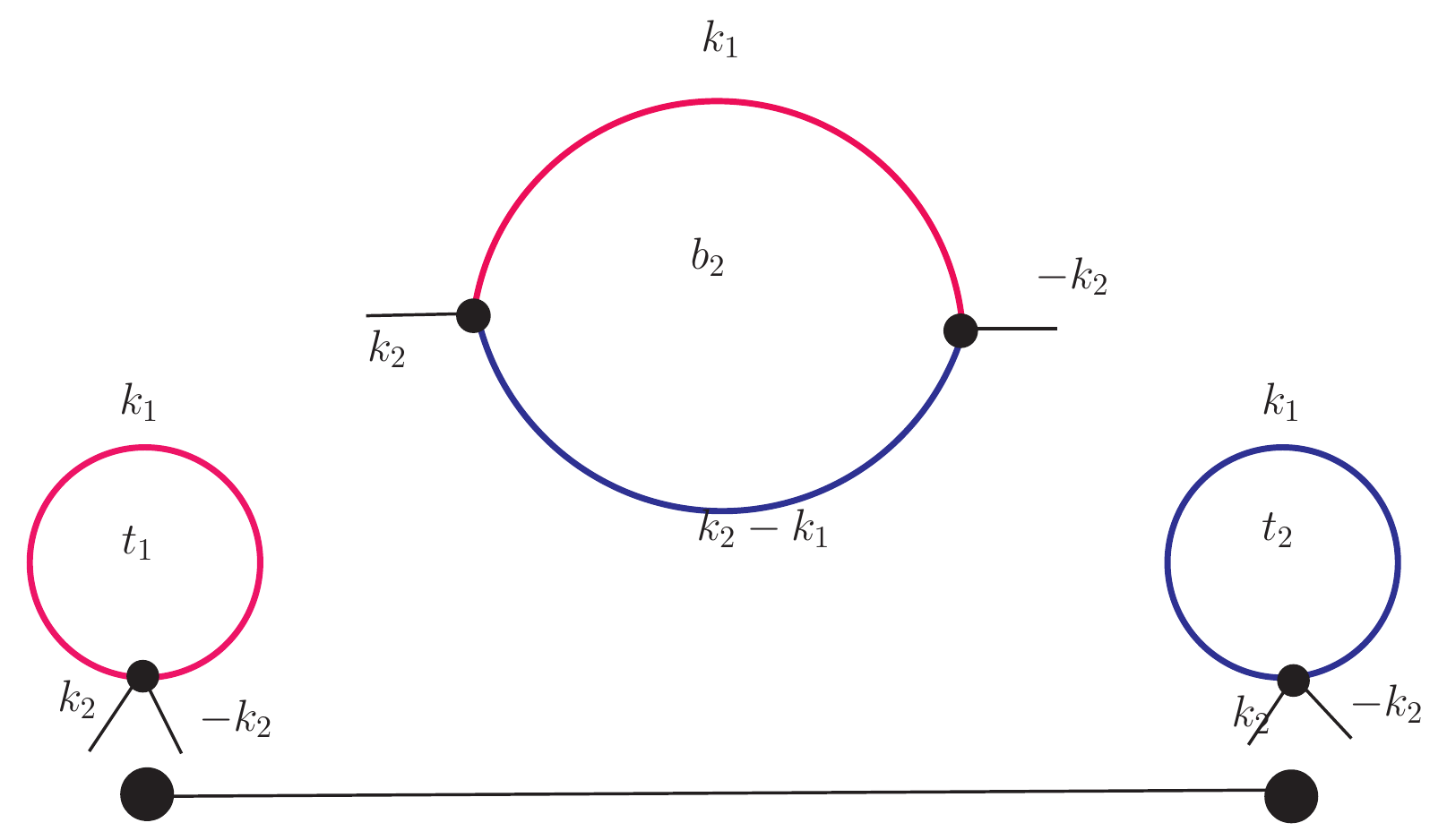}.
\caption{The bubble $b_2$. It gives rise to a function $\Phi_R^D(b_2)(k_2^2,m_1^2,m_2^2)$. We compute its imaginary part $\Im\left(\Phi_R^D(b_2)(k_2^2,m_1^2,m_2^2)\right)$ below. It starts an induction leading to the desired iterated integral for $\Im(\Phi_R^D(b_n))$. The edges $e_1,e_2$ are given in red or blue. Shrinking one of them gives a tadpole integral $\Phi_R^D(t_1)(m_1^2)$ (red) or $\Phi_R^D(t_2)(m_2^2)$ (blue).}
\label{btwo}
\end{figure}
The incoming external momenta at the two vertices of $b_2$  are $k_2,-k_2$
which can be regarded as momenta assigned to leaves at the two three-valent vertices.

We discuss the computation of $b_2$ in detail as it gives a start of an induction which leads to the computation of $b_n$. The underlying recursion goes long way back as discussed in Se.(\ref{history}) above, see \cite{Prem} in particular.  More precisely it allows to express
$\Im(\Phi_R^D)(b_n)$ as an iterated integral with the integral 
$\Im(\Phi_R^D)(b_2)$ as the start so that $b_n$ is obtained as a $(n-2)$-fold iterated one-dimensional integral.

For the Feynman integral $\Phi_R^D(b_2)$  we implement a kinematic renormalization scheme by subtraction at $s_0\equiv \mu^2\lneq (m_1-m_2)^2$ in accordance with Eq.(\ref{szero}). This implies that the subtracted terms do not have an imginary part, as $\mu^2$ is below the pseudo threshold $(m_1-m_2)^2$. For example for $D=4$
\[
\Phi_R^4(b_2)(s,s_0,m_1^2,m_2^2)=\int d^Dk_1 \left( \frac{1}{\underbrace{k_1^2-m_1^2}_{Q_1}}\frac{1}{\underbrace{(k_2-k_1)^2-m_2^2}_{Q_2}
}- \{k_2^2\to \mu^2\}\right).
\]
We have $s:=k_2^2$. For $D=6,8,\ldots$, subtractions of further Taylor coefficients at $s=\mu^2$ are needed.

As  the $D$-vector $k_2$ is assumed timelike (as $s>0$) we can work in a coordinate system where  $k_2=(k_{2;0},\vec{0})^T$ and get 
\[
\Phi_R^D(b_2)=\omega_{\frac{D}{2}} \int_{-\infty}^\infty dk_{1;0}\int_0^\infty \sqrt{t_1}^{D-3}dt_1 \left( \frac{1}{k_{1;0}^2-t_1-m_1^2}\frac{1}{(k_{2;0}-k_{1;0})^2-t-m_2^2}- \{s\to s_0\}\right).
\]
We define the K\.{a}llen function, actually a homogeneous polynomial, 
\[\lambda(a,b,c):=a^2+b^2+c^2-2(ab+bc+ca),\] and find by explicit integration,
for example for $D=4$,
\beas
& & \Phi_R^4(b_2)(s,s_0;m_1^2,m_2^2) = \\
& &  = \left(\underbrace{ \frac{\sqrt{\lambda(s,m_1^2,m_2^2)}}{2s}\ln\frac{m_1^2+m_2^2-s-\sqrt{\lambda(s,m_1^2,m_2^2)}}{m_1^2+m_2^2-s+\sqrt{\lambda(s,m_1^2,m_2^2)}} - \frac{m_1^2-m_2^2}{2s}\ln\frac{m_1^2}{m_2^2}}_{W_2^4(s)}\right. \\ & & \left. -\underbrace{\{s\to s_0\}}_{W_2^4(s_0)}\right). 
\eeas
The principal sheet of the above logarithm is  real for $s\leq (m_1+m_2)^2$ and free of singularities at $s=0$ and $s=(m_1-m_2)^2$. It has a branch cut for $s\geq (m_1+m_2)^2$. See for example \cite{RemSchouten,DirkEll} for a discussion of its analytic structure and behaviour off the principal sheet.

The threshold divisor defined by the intersection $L_1\cap L_2$ where the zero locii 
\[
L_i:\,Q_i=0,
\]
of the two quadrics meet is at $s=(m_1+m_2)^2$. This is an elementary example of the application of Picard--Lefshetz theory \cite{BlKrCut}.

Off the principal sheet, we have a pole at $s=0$ and a further branch cut for $s\leq (m_1-m_2)^2$.

It is particularly interesting to compute the variation -the imaginary part- of $\Phi_R(b_2)$ using Cutkosky's theorem \cite{BlKrCut}. For all $D$,
\[
\Im(\Phi_R^{D}(b_2))=\omega_{\frac{
D}{2}}\int_{0}^\infty \sqrt{t_1}^{D-3}dt \int_{-\infty}^\infty dk_{1;0} 
\delta_+(k_{1;0}^2-t_1-m_1^2)\delta_+((k_{2;0}-k_{1;0})^2-t_1-m_2^2).
\]
We have
\[
\delta_+((k_{2;0}-k_{1;0})^2-t_1-m_2^2)=\Theta(k_{2;0}-k_{1;0})\delta((k_{2;0}-k_{1;0})^2-t_1-m_2^2),
\]
and
\beas
\delta((k_{2;0}-k_{1;0})^2-t_1-m_2^2) & = & \frac{1}{2|k_{2;0}-k_{1;0}|}_{|k_{1;0}=k_{2;0}+\sqrt{t_1+m_2^2}}\times \delta(k_{1;0}-k_{2;0}-\sqrt{t_1+m_2^2})\\
 & + & \frac{1}{2|k_{2;0}-k_{1;0}|}_{|k_{1;0}=k_{2;0}-\sqrt{t_1+m_2^2}}\times \delta(k_{1;0}-k_{2;0}+\sqrt{t_1+m_2^2}).
\eeas
In summary
\beas
\delta_+((k_{2;0}-k_{1;0})^2-t_1-m_2^2) & = & \Theta(k_{2;0}-k_{1;0})\delta((k_{2;0}-k_{1;0})^2-t_1-m_2^2)\\
 & = &
\frac{1}{2|k_{2;0}-k_{1;0}|}_{|k_{1;0}=k_{2;0}-\sqrt{t_1+m_2^2}}\delta(k_{1;0}-k_{2;0}+\sqrt{t_1+m_2^2}),
\eeas
and therefore
\[
\Im(\Phi_R(b_2))=\omega_{\frac{D}{2}}\int_0^\infty \sqrt{t_1}^{D-3}dt_1 
\delta\left(s-2\sqrt{s}\sqrt{t_1+m_2^2}+m_2^2-m_1^2\right)\frac{1}{\sqrt{t_1+m_2^2}}.
\]
We have from the remaining $\delta$-function,
\[
\delta\left(s-2\sqrt{s}\sqrt{t_1+m_2^2}+m_2^2-m_1^2\right)=\frac{\sqrt{t_1+m_2^2}}{\sqrt{s}}
\delta\left(t_1-\frac{\lambda(s,m_1^2,m_2^2)}{4s}\right),
\]
hence
\[
0\leq t_1 =\frac{\lambda(s,m_1^2,m_2^2)}{4s},
\]
whenever the K\.{a}llen function $\lambda(s,m_1^2,m_2^2)$ is positive, so for $s>(m_1+m_2)^2$ (normal threshold, on the principal sheet) or for $0<s<(m_1-m_2)^2$ (pseudo-threshold, off the principal sheet).

The integral then gives
\[
\Im(\Phi_R^D(b_2))(s,m_1^2,m_2^2)=\overbrace{\omega_{\frac{D}{2}}\left( \frac{\left(\sqrt{\lambda(s,m_1^2,m_2^2)}\right)^{D-3}}{(2s)^{\frac{D}{2}-1}}\right) }^{=:V_{2}^{D}(s;m_1^2,m_2^2)}\times \Theta(s-(m_1+m_2)^2),
\]
with $\omega_{\frac{D}{2}}$ given in Eq.(\ref{omegadhalf}). We emphasize that $V_{2}^D$ has a pole at $s=0$ with residue $|m_1^2-m_2^2|/2$ and 
note  $\lambda(s,m_1^2,m_2^2)=(s-(m_1+m_2)^2)(s-(m_1-m_2)^2)$.

We regain $\Phi_R^D(b_2)$ from $\Im(\Phi_R^D(b_2))$ by a subtracted dispersion integral, for example for $D=4$:
\[
\Phi_R^4(b_2)(s,s_0)=\frac{s-s_0}{\pi}\int_0^\infty \frac{\Im\left(\Phi_R^4(b_2)\right)(x)}{(x-s)(x-s_0)}dx.
\]
Here, the renormalization condition implemented in the once-subtracted dispersion imposes $\Phi_R^D(b_2)(s_0,s_0)=0$ for $D=4$. 

Finally we note that for on-shell edges $(k_2-k_1)^2=m_2^2$ so
\beas
k_2\cdot k_1 & = & \frac{k_2^2-m_2^2+m_1^2}{2},\\
k_1^2 & = & m_1^2.
\eeas
\subsection{Computing $b_3$}
\label{sec:9}
\vspace{1mm}\noindent

We now consider the 3-edge banana $b_3$ on three different masses.
$$\bthree$$
We start by using the fact that we can disassemble $b_3$ in three different ways into a $b_2$ sub-graph, with a remaining edge providing the co-graph.  
Using Fubini, the three equivalent ways to write it in accordance with the flag structure $b_2\subset b_3$ are:
\bea
\Im(\Phi_R^D(b_3)) & = & \int d^Dk_2 \Im(\Phi_R^D(b_2))(k_2^2,m_1^2,m_2^2)
\delta_+((k_3-k_2)^2-m_3^2),\label{bthreeo}\\
\Im(\Phi_R^D(b_3)) & = &\int d^Dk_2 \Im(\Phi_R^D(b_2))(k_2^2,m_2^2,m_3^2)
\delta_+((k_3-k_2)^2-m_1^2),\label{bthreetw}\\
\Im(\Phi_R^D(b_3)) & = & \int d^Dk_2 \Im(\Phi_R^D(b_2))(k_2^2,m_3^2,m_1^2)
\delta_+((k_3-k_2)^2-m_2^2).\label{bthreeth}
\eea
In any of these cases for $\Im(\Phi_R^D(b_3))$ we integrate over the common support
of the distributions  
\[
\Im(\Phi_R^D(b_2))(k_2^2,m_i^2,m_j^2)\sim
\Theta(k_2^2-(m_i+m_j)^2)\,\,{\text{and}}\,\,\delta_+((k_3-k_2)^2-m_k^2),
\]
 generalizing the situation 
for $\Im(\Phi_R^D(b_2))$ where we integrated over the common support of 

\[
\delta_+(k_1^2-m_1^2)\,\,{\text{and}}\,\,\delta_+((k_2-k_1)^2-m_2^2).
\]

The integrals Eqs.(\ref{bthreeo},\ref{bthreetw},\ref{bthreeth}) are well-defined  and on the principal sheet they are equal and give the variation (and hence imaginary part) $\Im(\Phi_R^D(b_3))$ of $\Phi_R^D(b_3)$.

$\Phi_R^D(b_3)$ itself can be obtained from it by a sufficiently subtracted dispersion integral which reads for $D=4$
\[
\Phi_R^4(b_3)(s,s_0)=\frac{(s-s_0)^2}{\pi}\int_0^\infty \frac{\Im(\Phi_R^4(b_3)(x))}{(x-s)(x-s_0)^2}dx.
\]
For general $D$,
$\Phi_R^D(b_3)$
is well-defined no matter which of the two edges we choose as the sub-graph,
and Cutkosky's theorem defines a unique function $V_{3}^D(s)$,
\[
\Im(\Phi_R^D(b_3)(s))=:V_{3}^D(s)\Theta(s-(m_1+m_2+m_3)^2).
\]
\begin{rem}
Below when we discuss master integrals  for $b_n$ we find that by breaking symmetry through a derivative $\partial_{m_i^2}$ we obtain four master integrals for $b_3$.
$\Phi_R^D(b_3)$ itself,  and by applying $\partial_{m_i^2}$ to any of Eqs.(\ref{bthreeo},\ref{bthreetw},\ref{bthreeth}).
\end{rem}\hfill $|$

Let us compute $V_3^D$ first. We consider edges $e_1,e_2$ as a $b_2$ subgraph with an external momentum 
$k_2$ flowing through.

We let $k_3$ be the external momentum of $\Im(\Phi_R^D(b_3))$, $0<k_3^2=:s$. 
For the $k_2$-integration we put ourselves in the restframe $k_3=(k_{3;0},\vec{0})^T$.

Consider then
\[
\Im\left(\Phi_R^D(b_3)\right)(s)=\int d^Dk_2 \Theta(k_2^2-(m_1+m_2)^2)
\delta_+((k_3-k_2)^2)-m_3^2)V_2^D(k_2^2,m_1^2,m_2^2).\]
The $\delta_+$-distribution demands that $k_{3;0}-k_{2;0}>0$, and therefore we get
\beas
\Im\left(\Phi_R^D(b_3)\right)(s) & = & \omega_{\frac{D}{2}}\int_{-\infty}^{k_{3;0}} dk_{2;0}\int_0^\infty dt_2\sqrt{t_2}^{D-3}\Theta(k_{2;0}^2-t_2-(m_1+m_2)^2)\times\\
 & \times & V_2^D(k_{2;0}^2-t,m_1^2,m_2^2) \delta((k_{3;0}-k_{2;0})^2-t_2-m_3^2).
\eeas
As a function of $k_{2;0}$, the argument of the $\delta$-distribution has two zeroes:
\[
k_{2;0}=k_{3;0}\pm\sqrt{t_2+m_3^2}.
\]

As $k_{3;0}-k_{2;0}>0$, it follows $k_{2;0}=k_{3;0}-\sqrt{t_2+m_3^2}$.
Therefore, $k_{2;0}^2-t_2=k_{3;0}^2+m_3^2-2k_{3;0}\sqrt{t_2+m_3^2}$.

For our desired integral, we get
\beas
\Im\left(\Phi_R^D(b_3)\right)(s) & = & \omega_{\frac{D}{2}}\int_0^\infty dt_2 \sqrt{t_2}^{D-3} \Theta(k_{3;0}^2+m_3^2-2k_{3;0}\sqrt{t_2+m_3^2}-(m_1+m_2)^2)\times
\\
 & \times & \frac{ V_2^D\left(k_{3;0}^2+m_3^2-2k_{3;0}\sqrt{t_2+m_3^2},m_1^2,m_2^2\right)}{\sqrt{t_2+m_3^2}}.
\eeas
 The $\Theta$-distribution requires
 \[
 k_{3;0}^2+m_3^2-(m_1+m_2)^2\geq 2k_{3;0}\sqrt{t_2+m_3^2}.
 \]
Solving for $t_2$, we get
\[
0\leq t_2\leq \frac{\lambda(s,m_3^2,(m_1+m_2)^2)}{4s}.
\]
As $t_2\geq 0$, we must have for the physical threshold $s>(m_3+m_1+m_2)^2$ which is indeed completely symmetric under permutations of $1,2,3$, in accordance with our expectations for $\Im(\Phi_R^D(b_3)(s))$.
We then have
\beas
\Im(\Phi_R^D(b_3)(s)) & = & \Theta(s-(m_1+m_2+m_3)^2)\times\\
 & \times & \omega_{\frac{D}{2}}\int_0^{\frac{\lambda(s,m_3^2,(m_1+m_2)^2)}{4s}}
\frac{V_2^D(s+m_3^2-2\sqrt{s}\sqrt{t_2+m_3^2},m_1^2,m_2^2)}{\sqrt{t_2+m_3^2}}
\sqrt{t_2}^{D-3}dt_2.
\eeas
There is also a pseudo-threshold off the principal sheet at $s<(m_3-m_1-m_2)^2$, see Sec.(\ref{apppseudo}).

Note that the integrand vanishes at the upper boundary $\frac{\lambda(s,m_k^2,(m_i+m_j)^2)}{4s}$ as
\[
\lambda(s+m_3^2-2\sqrt{s}\sqrt{t_2+m_3^2},m_1^2,m_2^2)_{\mid t_2=\frac{\lambda(s,m_3^2,(m_1+m_2)^2)}{4s}}=\lambda((m_1+m_2)^2,m_1^2,m_2^2)=0.
\]
Let us now transform variables.
\beas
y_2 & := & \sqrt{t_2+m_{3}^2},\\
t_2 & = & y_2^2-m_{3}^2,\\
dt_2 & = & 2y_2 dy_2,\\
\int_0^{\frac{\lambda}{4s}} & \to & \int_{m_{3}}^{\frac{s+m_{3}^2-(m_1+m_2)^2}{2\sqrt{s}}}. 
\eeas
We get
\bea\label{vthree}
\Im(\Phi_R^D(b_3)(s)) & = & \Theta(s-(m_1+m_2+m_3)^2)\times\nonumber\\
 & \times & \underbrace{\omega_{\frac{D}{2}}\int_{m_3}^{\frac{s+m_{3}^2-(m_1+m_2)^2}{2\sqrt{s}}}
V_2^D\left(\overbrace{s+m_3^2-2\sqrt{s}y_2}^{s_3^1(y_2,m_3^2)},m_1^2,m_2^2\right)\sqrt{y_2-m_3^2}^{D-3}
dy_2}_{V_3^{D}(s,m_1^2,m_2^2,m_3^2)}.
\eea
Had we choosen $e_2,e_3$ or $e_3,e_1$ instead of $e_1,e_2$ for $b_2$ we would find in accordance with Eqs.(\ref{bthreeo},\ref{bthreetw},\ref{bthreeth})
\bea\label{vthreetw}
\Im(\Phi_R^D(b_3)(s)) & = & \Theta(s-(m_1+m_2+m_3)^2)\times\nonumber\\
 & \times & \underbrace{\omega_{\frac{D}{2}}\int_{m_1}^{\frac{s+m_{1}^2-(m_2+m_3)^2}{2\sqrt{s}}}
V_2^D\left(\overbrace{s+m_1^2-2\sqrt{s}y_2}^{s_3^1(y_2,m_1^2)},m_2^2,m_3^2\right)\sqrt{y_2-m_1^2}^{D-3}
dy_2}_{V_3^{D}(s,m_1^2,m_2^2,m_3^2)},
\eea
or
\bea\label{vthreeth}
\Im(\Phi_R^D(b_3)(s)) & = & \Theta(s-(m_1+m_2+m_3)^2)\times\nonumber\\
 & \times & \underbrace{\omega_{\frac{D}{2}}\int_{m_2}^{\frac{s+m_{2}^2-(m_3+m_1)^2}{2\sqrt{s}}}
V_2^D\left(\overbrace{s+m_2^2-2\sqrt{s}y_2}^{s_3^1(y_2,m_2^2)},m_3^2,m_1^2\right)\sqrt{y_2-m_2^2}^{D-3}
dy_2}_{V_3^{D}(s,m_1^2,m_2^2,m_3^2)}.
\eea
with three different $s_3^1(y_2)=s_3^1(y_2,m_i^2)$.

We omit this distinction in the future as we will always choose a fixed order of edges and call the edges in the innermost bubble $b_2$ edges  $e_1,e_2$.

Finally, we note
\beas
k_{2,0} & = & k_{3,0}-y_2,\\
k_2^2 & = & k_{3,0}^2-2k_{3,0}y_2+m_3^2,\\
|\vec{k_2}| & = & \sqrt{y_2^2-m_3^2}.
\eeas
Written in invariants this is 
\beas
k_3\cdot k_{2} & = & \sqrt{s}(\sqrt{s}-y_2),\\
k_2^2 & = & s-2\sqrt{s}y_2+m_3^2,\\
|\vec{k_2}| & = & \sqrt{y_2^2-m_3^2}.
\eeas
\subsection{$b_3$ and elliptic integrals}\label{bthreeelliptic}
Note that for $D=2$ (the case $D=4$ can be treated similarly as in \cite{RemSchouten}) and using Eq.(\ref{vthree}),
\[
V_3^2(s)=
\omega_{1}\int_{m_3}^{\frac{s+m_{3}^2-(m_2+m_1)^2}{2\sqrt{s}}}
\frac{1}{\sqrt{U(y_2)}}
dy_2,
\]
with
\[
U(y_2)=\lambda\left({s+m_3^2-2\sqrt{s}y_2},m_2^2,m_1^2\right)(y_2^2-m_3^2)
=s(y_2-m_3)(y_2+m_3)(y_2-y_+)(y_2-y_-),
\]
a quartic polynomial so that $V_3^2$ defines an elliptic integral following for example \cite{RemSchouten}.
Here,
\[
y_\pm=\frac{(s+m_3^2-m_1^2-m_2^2)\pm 2\sqrt{m_1^2m_2^2}}{2\sqrt{s}}.
\]
So indeed
\be\label{bthreeconcrete}
V_3^2(s)=\frac {2\omega_1}{(y_++m_3)(y_--m_3)}K\left(\frac{(y_-+m_3)(y_+-m_3)}{(y_--m_3)(y_++m_3)}\right),
\ee
with $K$ the complete elliptic integral of the first kind.

Finally 
\be\label{bthreedeqtwo}
\Phi_R^2(b_3)(s)=\frac{1}{\pi}\int_{(m_1+m_2+m_3)^2}^\infty \frac{V_3^2(x)}{(x-s)} dx,
\ee
gives the full result for $b_3$ in terms of elliptic dilogarithms in all its glory \cite{Weinzierl,BlochVanhove,Vanhove} for $D=2$. 
For arbitrary $D$ we get
\be\label{bthreed}
\Phi_R^D(b_3)(s,s_0)=\frac{(s-s_0)^{D-2}}{\pi}\int_{(m_1+m_2+m_3)^2}^\infty \frac{V_3^D(x)}{(x-s)(x-s_0)^{D-2}} dx.
\ee

To compare our result Eq.(\ref{bthreeconcrete}) with the result in \cite{RemSchouten} say, note that we can write 
\[
U(y_2)=\frac{1}{4}\lambda(s,s_3^1,m_3^2)\lambda(s_3^1,m_1^2,m_2^2),
\]
as 
\[
\lambda(s,s_3^1,m_3^2)=(s_3^1-(\sqrt{s}-m_3)^2)(s_3^1-(\sqrt{s}+m_3)^2)=4s(y_2^2-m_3^2),
\]
with $s_3^1=s-2\sqrt{s}y_2+m_3^2$, and use $b=s_3^1$, $db=-2\sqrt{s}dy_2$
to compare.

\subsection{Computing $b_4$}\label{bfour}
Above we have expressed $V_3^D$ as an integral involving $V_2^D$.
We can iterate this procedure.

Let us compute $V_4^D$ next repeating the computation which led to Eq.(\ref{vthree}). We consider edges $e_1,e_2,e_3$ as a $b_3$ subgraph with an external momentum 
$k_3$ flowing through.

We let $k_4$ be the external momentum of $\Im(\Phi_R^D(b_4))$, $0<k_4^2=s$. 
We put ourselves in the restframe $k_4=(k_{4;0},\vec{0})^T$ for the $k_3$-integration.

Consider then
\[
\Im\left(\Phi_R^D(b_4)\right)(s)=\int d^Dk_3 \Theta(k_3^2-(m_1+m_2+m_3)^2)
\delta_+((k_4-k_3)^2)-m_4^2)V_3^D(k_3^2,m_1^2,m_2^2,m_3^2).
\]
The $\delta_+$ distribution demands that $k_{4;0}-k_{3;0}>0$, and therefore we get
\beas
\Im\left(\Phi_R^D(b_4)\right)(s)=\omega_{\frac{D}{2}}\int_{-\infty}^{k_{4;0}} dk_{3;0}\int_0^\infty dt_3\sqrt{t_3}^{D-3}\Theta(k_{3;0}^2-t_3-(m_1+m_2+m_3)^2) & & \\V_3^D(k_{3;0}^2-t_3,m_1^2,m_2^2,m_3^2) \delta((k_{4;0}-k_{3;0})^2-t_3-m_4^2).& & 
\eeas
As a function of $k_{3;0}$, the argument of the $\delta$-distribution has two zeroes:
$k_{3;0}=k_{4;0}\pm\sqrt{t_3+m_4^2}$.

As $k_{4;0}-k_{3;0}>0$, it follows $k_{3;0}=k_{4;0}-\sqrt{t_3+m_4^2}$.
Therefore, $k_{3;0}^2-t_3=k_{4;0}^2+m_4^2-2k_{4;0}\sqrt{t_3+m_4^2}$.

For our desired integral, we get
\beas
\Im\left(\Phi_R^D(b_4)\right)(s) & = & \omega_{\frac{D}{2}}\int_0^\infty dt_ 3 \sqrt{t_3}^{D-3} \Theta(k_{4;0}^2+m_4^2-2k_{4;0}\sqrt{t_3+m_4^2}-(m_1+m_2+m_3)^2)\times\\
 & \times &  \frac{ V_3^D\left(k_{4;0}^2+m_4^2-2k_{4;0}\sqrt{t_3+m_4^2},m_1^2,m_2^2,m_3^2\right)}{\sqrt{t_3+m_4^2}}.
\eeas
 The $\Theta$-distribution requires
 \[
 k_{4;0}^2+m_4^2-(m_1+m_2+m_3)^2\geq 2k_{4;0}\sqrt{t_3+m_4^2}.
 \]
Solving for $t_3$, we get
\[
0\leq t_3\leq \frac{\lambda(s,m_4^2,(m_1+m_2+m_3)^2)}{4s}.
\]
As $t_3\geq 0$, we must have for the physical threshold $s>(m_4+m_3+m_1+m_2)^2$.
We then have
\beas
\Im(\Phi_R^D(b_4)(s)) & = & \Theta(s-(m_1+m_2+m_3+m_4)^2)\times\\
 & \times & \omega_{\frac{D}{2}}\int_0^{\frac{\lambda(s,m_4^2,(m_1+m_2+m_3)^2)}{4s}}
\frac{V_3^D(s+m_4^2-2\sqrt{s}\sqrt{t_3+m_4^2},m_1^2,m_2^2,m_3^2)}{\sqrt{t_3+m_4^2}}
\sqrt{t_3}^{D-3}dt_3.
\eeas

Let us now transform variables again.
\beas
y_3 & := & \sqrt{t_3+m_{4}^2},\\
t_3 & = & y_3^2-m_{4}^2,\\
dt_3 & = & 2y_3 dy_3,\\
\int_0^{\frac{\lambda}{4s}} & \to & \int_{m_{4}}^{\frac{s+m_{4}^2-(m_1+m_2+m_3)^2}{2\sqrt{s}}}. 
\eeas
We get
\beas
\Im(\Phi_R^D(b_4)(s)) & = & \Theta(s-(m_1+m_2+m_3+m_4)^2)\times\\
 & \times & \underbrace{\omega_{\frac{D}{2}}\int_{m_4}^{\frac{s+m_{4}^2-(m_1+m_2+m_3)^2}{2\sqrt{s}}}
V_3^D(\overbrace{s+m_4^2-2\sqrt{s}y_3}^{s_4^1(y_3)},m_1^2,m_2^2,m_3^2)\sqrt{y_3-m_4^2}^{D-3}
dy_3}_{V_4^{D}(s,m_1^2,m_2^2,m_3^2,m_4^2)}.
\eeas
We have thus expressed $V_4^D$ as an integral involving $V_3^D$.
As we can express $V_3^D$ by $V_2^D$, we get the iterated integral,
\bea\label{vfour}
V_4^D(s,m_1^2,m_2^2,m_3^2,m_4^2) & = & 
\omega_{\frac{D}{2}}^2\int_{m_4}^{\frac{s+m_{4}^2-(m_1+m_2+m_3)^2}{2\sqrt{s}}}
\Biggl(\int_{m_3}^{\frac{s_4^1(y_3)+m_{3}^2-(m_1+m_2)^2}{2\sqrt{s_4^1(y_3)}}}\times
\nonumber\\
 & \times  & V_2^D(s_4^2(y_2,y_3),m_1^2,m_2^2)\sqrt{y_2-m_3^2}^{D-3}dy_2\Biggr)\sqrt{y_3-m_4^2}^{D-3}
dy_3.
\eea
We abbreviated 
\[
s_4^2(y_2,y_3):=s_4^1(y_3)-2\sqrt{s_4^1(y_3)}y_2+m_3^2=s_4^0-2\sqrt{s_4^0}y_3+m_4^2-2\sqrt{s_4^0-2\sqrt{s_4^0}y_3+m_4^2}y_2+m_3^2,
\] 
$s_4^0:=s$.
\subsection{Beyond elliptic itegrals for $b_4$}\label{bfournonell}
Note that $V_4^2$ can not be read as a complete elliptic integral of any kind. 
It is a double integral over the inverse square root of an algebraic function.
$V_3^2$ was in contrast a single integral over the inverse square root of a mere quartic polynomial.
Concretely the relevant integrand is
\[
\frac{1}{\sqrt{(y_3^2-m_4^2)^2(y_2^2-m_3^2)v_4(y_2,y_3)}}.
\]

In fact the innermost $y_2$ integral can still be expressed as a complete elliptic 
integral of the first kind as in Eq.(\ref{bthreeconcrete}), as $v_4$ is a quadratic polynomial in $y_2$ so that 
\[
(y_2^2-m_3^2)v_4=(y_2-m_3)(y_2+m_3)(y_2-y_{2,+})(y_2-y_{2,-})
\]
 is a quartic in $y_2$ albeit with coefficients $y_{2,\pm}$ which are algebraic in $y_3$.
We have
\[
y_{2,\pm}(y_3)=\frac{(m_1^2+m_2^2-m_3^2-s_4^1(y_3))
\pm 2\sqrt{m_1^2m_2^2}}{2 \sqrt{s_4^1(y_3)}}.
\]
We get the more than elliptic integral over an elliptic integral of the first kind,
\bea\label{bfourythree}
V_4^2(s) & = & \omega_1\int_{m_4}^{\frac{s+m_{4}^2-(m_1+m_2+m_3)^2}{2\sqrt{s}}}\frac {2\omega_1}{(y_{2,+}(y_3)+m_4)(y_{2,-}(y_3)-m_4)}\times\nonumber\\
 & \times & K\left(\frac{(y_{2,-}(y_3)+m_4)(y_{2,+}(y_3)-m_4)}{(y_{2,-}(y_3)-m_4)(y_{2,+}(y_3)+m_4)}\right) \frac{1}{\sqrt{y_3^2-m_4^2}}dy_3.
\eea

\subsection{Computing $b_n$ by iteration}\label{bn}
Iterating the computation which led to Eq.(\ref{vfour}) we get
\begin{thm}\label{monodromyThm}
Let $b_n$ be the banana graph on $n$ edges and two leaves (at two distinct vertices) with masses $m_i$ and momenta $k_n,-k_n$ incoming at the two vertices in $D$ dimensions.\\
i) it has an imaginary part determined by a normal threshold as
\[
\Im\left(\Phi_R^D(b_n)\right)(s)=\Theta\left(s-\left(\sum_{j=1}^n m_j\right)^2\right)V_n^D(s,\{m_i^2\}),
\]
and with a recursion ($n\geq 3$)
\beas
V_n^D(s;\{m_i^2\})
 & = & \omega_{\frac{D}{2}}\int_{m_n}^{\frac{s+m_n^2-(\sum_{j=1}^{n-1}m_j)^2}{2\sqrt {s_n^0}}}V_{n-1}^D(s_n^0-2\sqrt{s_n^0}y_{n-1}+m_n^2,m_1^2,\ldots,m_{n-1}^2)\times\\
 & \times & \sqrt{y_{n-1}^2-m_n^2}^{D-3}dy_{n-1}.
\eeas
\end{thm}
\noindent\textbf{Rem. i)}\\
This imaginary part is the variation in $s$ of $\Phi_R^D(b_n)(s)$ in the principal sheet. Variations on other sheets are collected in App.(\ref{apppseudo}). See \cite{DirkEll} for an introduction to a discussion of the r\^ole of such pseudo-thresholds.\hfill$|$
\begin{thm*}(cont'd)\\
ii) Define for all $n\geq 2$, $0\leq j\leq n-2$, 
\[
s_n^0:=s,
\] 
and for $n-2\geq j\geq 1$, 
$s_n^j=s_n^j(y_{n-j},\ldots,y_{n-1};m_n,\ldots,m_{n-j+1})$,
\be\label{defsn}
s_n^j
=s_n^{j-1}-2\sqrt{s_n^{j-1}}y_{n-j}+m^2_{n-j+1}.
\ee\label{defup}
Define 
\be\label{up}
\mathrm{up}_n^j:=\frac{s_n^j+m_{n-j}^2-\left(\sum_{i=1}^{n-j-1}m_i\right)^2}{2\sqrt{s_n^{j}}},
\ee
then $V_n^D$ is given by the following iterated integral:
\bea\label{itInt}
V_n^D(s,m_1^2,\ldots,m_n^2) & := & \omega_{\frac{D}{2}}^{n-2}\int_{m_n}^{\mathrm{up}_n^0}\Biggl(\int_{m_{n-1}}^{\mathrm{up}_n^1(y_{n-1})}\Biggl(
\int_{m_{n-2}}^{\mathrm{up}_n^2(y_{n-1},y_{n-2})}\cdots\nonumber\\
 & \cdots & \Biggl(\int_{m_3}^{\mathrm{up}_n^{n-3}(y_3,\ldots,y_{n-1})}
V_2^D(s_n^{n-2}(y_2,\ldots,y_{n-1}),m_1^2,m_2^2)\times\nonumber\\
 & \times & 
\sqrt{y_2^2-m_{3}^2}^{D-3}dy_2 \Biggr)\cdots 
\sqrt{y_{n-2}^2-m_{n-1}^2}^{D-3} dy_{n-2}\Biggr)\times\\
 & \times & \sqrt{y_{n-1}^2-m_{n}^2}^{D-3} dy_{n-1}.\nonumber
\eea
Here, $V_2^D(a,b,c)=\frac{\lambda(a,b,c)^{\frac{D-3}{2}}}{a^{\frac{D}{2}-1}},$
so that
\[
V_2^D(s_n^{n-2}(y_2,\ldots,y_{n-1}),m_1^2,m_2^2)=\omega_{\frac{D}{2}}\frac{\lambda\Bigl(s_n^{n-2}(y_2,\ldots,y_{n-1}),m_1^2,m_2^2\Bigr)^{\frac{D-3}{2}}}{\Bigl(s_n^{n-2}(y_2,\ldots,y_{n-1})\Bigr)^{\frac{D}{2}-1}}.
\]
\end{thm*}
\noindent\textbf{Rem. ii)}\\
We solve the recursion in terms of an iteration of  one-dimensional integrals.
$V_2^D(b_2)$ serves as the seed, $V_2^D=\omega_{\frac{D}{2}}\lambda(s_n^{n-2},m_1^2,m_2^2)/s^{\frac{D}{2}-1}$) and $s_n^{n-2}=s_n^{n-2}(y_{n-1},\ldots,y_2;m_3^2,\ldots,m_n^2)$
depends on integration variables $y_j$ and on mass squares $m_{j+1}^2$, $j=2,\ldots,n-1$.
For $b_3$ we need a single integration, for $b_n$ we need to iterate $(n-2)$ integrals. Note that we could always do the innermost $y_2$-integral in terms of a complete elliptic integral (replacing $s_4^1\to s_n^{n-3}$ in Eq.(\ref{bfourythree}) etc.) and use that as the seed.\hfill$|$
\begin{thm*}(cont'd)\\
iii) We have the following identities:
\bea
V_n^D\left(\left(\sum_{j=1}^nm_j\right)^2;\{m_i^2\}\right) & = & 0,\label{idsone}\\
\mathrm{up}_n^1(y_{n-1})_{|y_{n-1}=\mathrm{up}_n^0} & = & m_{n-1},\label{idstwo}\\
\mathrm{up}_n^j(y_{n-j},\ldots, y_{n-1})_{|y_{n-j}=\mathrm{up}_n^{j-1}} & = & m_{n-j},\label{idsthree}\\
\mathrm{up}_n^{n-3}(y_{3},\ldots, y_{n-1})_{|y_{3}=\mathrm{up}_n^{n-4}} & = & m_3,
\label{idsfour}\\
V_2^D(s_n^{n-2},m_1^2,m_2^2)_{|y_{2}=\mathrm{up}_n^{n-3}} & = & 0.\label{idfive}
\eea
\end{thm*}
\noindent\textbf{Rem. iii)}\\
Eq.(\ref{idsone}) ensures that the dispersion integrand vanishes at the lower boundary $x=(m_1+\cdots +m_n)^2$ (the normal threshold) as it should.
Following  Eqs.(\ref{idstwo}-\ref{idsfour}) for any $y_j$-integration but the innermost integration the integrand vanishes at the lower and upper boundaries.
By Eq.(\ref{idfive}) for the innermost $y_{n-1}$ integral this  holds for $D\gneq 2$. 

At $D=2$ the result can be achieved by considering
\[
\lim_{\eta\to 0}\int_{m_3+\eta}^{\mathrm{up}_n^{n-3}-\eta}\cdots dy_{n-1}.
\]
In the limit $\sqrt{s}\to m_{\mathbf{normal}}^n$ for which $\mathrm{up}_n^{n-3}\to m_3$ one confirms the analysis in \cite{RemSchouten} that 
a finite value at threshold remains.

Summarizing  for any $D$ this amounts to compact integration as we have in any $y_j$ integration a resurrection of Stokes formula
\be\label{stokes} 
\int_{m_{j+1}}^{\mathrm{up}_n^{j+1}}\partial_{y_j}f(y_j)\cdots dy_{j}=0,
\ee
for any rational function $f(y_j)$ inserted as a coefficient of $V_2^D$.
The dots correspond to the other iterations of integrals in the $y_j$ variables. These are integration-by-parts identities.

This reflects the fact that the $n$
$\delta$-functions in a cut banana $b_n$ constrain the $(n-1)$ integrations of $k_{j;0}$, $j=1,\cdots,n-1$ and also the total integration over $r=\sum_{j=1}^{n-1}|\vec{k_j}|$. 
Here we can set $|\vec{k_j}|=r u_j$, and the $u_j$ parameterize a $(n-1)$-simplex
and hence a compactum. Angle integrals are over compact surfaces $S^{D-2}$. Only integrations over boundaries remain.\hfill$|$
\begin{thm*}(cont'd)\\
iv) We have
\be\label{partialykm}
\partial_{y_{k}} s_n^j=-2\sqrt{s_n^{n-k-1}}\partial_{m^2_{k+1}}s_n^j,\forall (n-j)\leq k\leq (n-1),
\ee
if all masses are different. The case of some equal masses is left to the reader.\\
Also,
\be\label{partials}
\left(\prod_{j=0}^{i-1}\sqrt{s_n^j}\right)\partial_{s}s_n^i=2\prod_{j=0}^{i-1}
\left(\sqrt{s_n^{j}}-y_{n-j-1}\right).
\ee
For derivatives with respect to masses we have for $0\leq r\lneq k-1$,
\be\label{partialm}
\partial_{m_{n-r}^2}s_n^k=\prod_{j=r+1}^{k-1}\frac{\sqrt{s_n^{j}}-y_{n-(j+1)}}{\sqrt{s_n^{j}}}.
\ee
whilst $\partial_{m_{n-k+1}^2}s_n^k=1$. 
Furthermore for $1\leq i\leq n-2-r$, $0\leq r\leq n-3$,
\be\label{partialyk}
\partial_{y_{n-i}}s_{n}^{n-2-r}=-2\sqrt{s_n^{i-1}}\prod_{j=2+r}^{n-1-i}\frac{s_n^{n-j-1}-y_{j}}{s_n^{n-j-1}}.
\ee
\end{thm*}
\noindent\textbf{Rem. iv)}\\
These formulae allow to trade $\partial_{y_j}$ derivatives with $\partial_{m_{j+1}^2}$ derivatives, and to treat $\partial_s$ derivatives. This is useful below when 
discussing differential equation, integration-by-parts  and master integrals for $\Phi_R^D(b_n)$.\hfill$|$
\begin{thm*}(cont'd)\\
v) Dispersion. Let $|[n,\nu]|-1$ (see Eq.(\ref{degdiv})) be the degree of divergence of $\Phi_R^D(b_n)_\nu$.
Then
\[
\Phi_R^D(b_n)_\nu(s,s_0)=\frac{(s-s_0)^{|[n,\nu]|}}{\pi}\int_{\left(\sum_{j=1}^n m_j\right)^2}^\infty
\frac{V_{[n,\nu]}^D(x,\{m_i^2\}}{(x-s)(x-s_0)^{|[n,\nu]|}}dx,
\]
is the renormalized banana graph with renormalization conditions
\[
\Phi_R^D(b_n)_\nu^{(j)}(s_0,s_0)=0,\, j\leq |[n,\nu]|-1,
\]
where $\Phi_R^D(b_n)_\nu^{(j)}(s_0,s_0)$ is the $j$-th derivative of $\Phi_R^D(b_n)_\nu(s,s_0)$ at $s=s_0$.
\end{thm*}
\noindent\textbf{Rem. v)}\\
This gives $\Phi_R^D(b_n)_\nu$ from $V_{[n,\nu]}^D$ in kinematic renormalization.
See App.(\ref{apptens}) for notation.
For a result in dimensional integration with MS use an unsubtracted dispersion
\[
\Phi_{MS}^D(b_n)_\nu(s)=\frac{1}{\pi}\int_{\left(\sum_{j=1}^n m_j\right)^2}^\infty
\frac{V_{[n,\nu]}^D(x,\{m_i^2\}}{(x-s)}dx,
\]
and then renormalize by Eq.(\ref{HopfBanana}) as tadpoles do not vanish in MS.\hfill$|$
\begin{thm*}(cont'd)\\
vi) Tensor integrals (see App.\ref{apptens}).
We have
\bea
k_{j+1}\cdot k_j & = & m_{j+1}^2-s_n^{n-j-1}-s_n^{n-j}=\label{tensors}\\
 & & =-\sqrt{s_n^{n-j-1}}\left(\sqrt{s_n^{n-j-1}}-y_{n-j}\right),\,j\geq 2,\nonumber\\
k_2\cdot k_{1} & = & \frac{k_2^2-m_2^2+m_1^2}{2},\label{tensorstwo}\\
k_j^2 & = & s_n^{n-j},\,\,{\text{in\,particular}} \,\,k_2^2=s_n^{n-2},\label{tensorsfour}\\
k_{j}\cdot k_l & = & \frac{k_l\cdot k_{l+1}k_{l+1}\cdot k_{l+2}\cdots k_{j-1}\cdot k_j}{{k_{l+1}^2}\cdots {k_{j-1}^2}}=\label{tensorsq}\\
 & & =\frac{\sqrt{s_n^{n-j-1}}}{\sqrt{s_n^{n-l-1}}}\prod_{i=l+1}^j \left(\sqrt{s_n^{n-i}}-y_{i+1}\right),\,j-l\gneq 1 ,\, j>l, l\gneq 1,\nonumber\\ 
k_{j}\cdot k_1 & = & \frac{k_l\cdot k_{l+1}k_{l+1}\cdot k_{l+2}\cdots k_{j-1}\cdot k_j}{{k_{l+1}^2}\cdots {k_{j-1}^2}}=\label{tensorsthree}\\
 & & =\frac{\sqrt{s_n^{n-j-1}}}{\sqrt{s_n^{n-2}}}\frac{s_n^{n-2}-m_2^2+m_1^2}{\sqrt{s_n^{n-2}}}\prod_{i=2}^j \left(\sqrt{s_n^{n-i}}-y_{i+1}\right),\,j-1\gneq 1.
\nonumber
\eea
Furthermore $V^D_{[n,\nu]}$ is obtained by
using Eqs.(\ref{tensors}-\ref{tensorsthree}) to insert tensor powers as indicated by $\nu$ in the integrand of $V_2^D(s_n^{n-2},m_1^2,m_2^2)$ and apply derivatives with respect to mass-squares accordingly. 
\end{thm*}
\noindent\textbf{Rem. vi)}\\
We first give in Fig.(\ref{bananastensor}) with $k_j^2=s_n^{n-j}$
also the irreducible squares of internal momenta (there is no propagator $k_j^2-m_j^2$ 
in the denominator of $b_n$).
\begin{figure}[H]
\includegraphics[width=12cm]{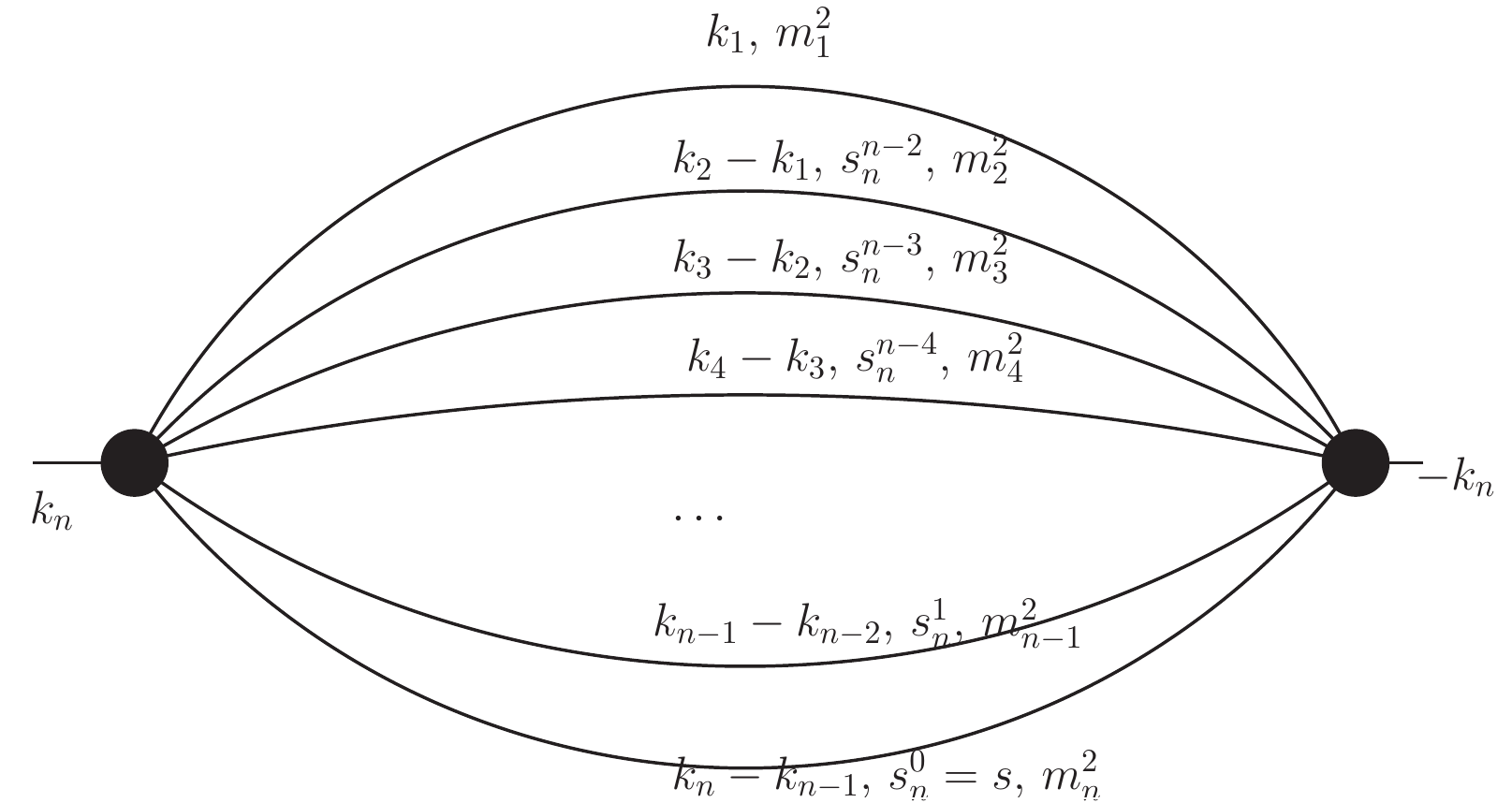}.
\caption{We indicate momenta 
and masses at internal edges from top to bottom. We now also indicate momentum $s_n^j$ for edges $e_2,\ldots, e_n$. The mass-shell conditions encountered in the computation of $V_n^D$ enforce $k_j^2=s_n^{n-j}$
for $2\leq j\leq n$. Eq.(\ref{tensors}) simply expresses the fact that 
$-2k_j\cdot k_{j+1}=(k_{j+1}-k_j)^2-k_{j+1}^2-k_j^2$ with $(k_{j+1}-k_j)^2=m_{j+1}^2$.}
\label{bananastensor}
\end{figure}

Eq.(\ref{tensorstwo}) is needed as Eq.(\ref{tensors}) can not cover the case $j=1$, due to the fact that for the $b_2$ integration $d^Dk_1$ both edges are constrained by a $\delta_+$-function, while each other loop integral gains only one more constraint, giving us a $y_j$ variable.

Eqs.(\ref{tensors}-\ref{tensorsthree}) allow to treat tensor integrals involving scalar products of irreducible numerators. Irreducible as there is no propagator $1/(k_j^2-m_{j+1}^2)$ in our momentum routing for $b_n$, see Fig.(\ref{bananastensor}).

Eqs.(\ref{tensorsq},\ref{tensorsthree}) for irreducible scalar products follow by integrating tensors in the numerator in the order of iterated integration.
For example
\[
\int\int  k_1\cdot k_3\frac{1}{\cdots}d^Dk_1d^Dk_2
=\int A(k_2^2)k_2\cdot k_3\frac{1}{\cdots} d^D k_2=C(k_3^2),
\] 
and
\[
\int\int  \frac{k_1\cdot k_2 k_2\cdot k_3}{k_2^2}\frac{1}{\cdots}d^Dk_1d^Dk_2
=\int A(k_2^2)\frac{k_2^2 k_2\cdot k_3}{k_2^2}\frac{1}{\cdots} d^D k_2=C(k_3^2),
\] 
using
\[
\int  \frac{{k_1}_\mu}{\cdots}d^Dk_1= A(k_2^2){k_2}_\mu,
\]
and dots $\cdots$ correspond to the obvious denominator terms.\hfill$|$
\begin{proof}
i) and ii) follow from the derivation of Eq.(\ref{vfour}) upon setting $4\to n$, $3\to n-1$ in an obvious manner.\\
iii) follows from inspection of Eq.(\ref{defup}):
For example
\[
\mathrm{up}_n^0=\frac{s+m_n^2-(m_1+\cdots+m_{n-1})^2}{2\sqrt{s}},
\]
\[
\mathrm{up}_n^1(y_{n-1})=\frac{s_n^1+m_{n-1}^2-(m_1+\cdots+m_{n-2})^2}{2\sqrt{s_n^1}},
\]
with
\[
s_n^1(y_{n-1})=s-2\sqrt{s}y_{n-1}+m_n^2.
\]
Then,
\[
\mathrm{up}_n^1(\mathrm{up}_n^0)=\frac{(m_1+\dots+m_{n-1})^2+m_{n-1}^2-(m_1+\cdots+m_{n-2})^2}{2(m_1+\cdots+m_{n-1})}=m_{n-1},
\]
and so on.\\
iv) straight from the definition Eq.(\ref{defsn}) of $s_n^j$. For example 
\[
\partial_{m_n^2} s_n^3=\frac{\left(\sqrt{s_n^1}-y_{n-2}\right)\left(\sqrt{s_n^2}-y_{n-3}\right)}{\sqrt{s_n^1}\sqrt{s_n^2}}
\]
v) This is the definiton of dispersion in kinematic renormalization conditions.\\
vi) For tensor integrals we collect variables $k_{j;0}$ and $t_j$ in any step of the computation in terms of $y_j=\sqrt{t_j+m_{j+1}^2}$.
\end{proof}
\subsubsection{$s_n^j$: iterating square roots}
Choose an order $o$ of the edges which fixes 
\[
b_2\subset b_3\subset\cdots\subset b_{n-1}\subset b_n.
\] 
Here we label
\[
E_{b_2}=: \{e_1,e_2\},\,E_{b_3}=\{e_1,e_2,e_3\},\ldots,E_{b_n}=\{E_{b_{n-1}}\cup e_n\}.
\]
Then
\beas
s_n^1(y_{n-1}) & = & s-2\sqrt{s}y_{n-1}+m_n^2,\\
s_n^2(y_{n-1},y_{n-2}) & = & s-2\sqrt{s}y_{n-1}+m_n^2-2\sqrt{s-2\sqrt{s}y_{n-1}+m_n^2}y_{n-2}+m_{n-1}^2,\\
s_n^3(y_{n-1},y_{n-2},y_{n-3}) & = & 
s-2\sqrt{s}y_{n-1}+m_n^2
  -  2\sqrt{s-2\sqrt{s}y_{n-1}+m_n^2}y_{n-2}+m_{n-1}^2\\
   & - & 2
\sqrt{s-2\sqrt{s}y_{n-1}+m_n^2
  -  2\sqrt{s-2\sqrt{s}y_{n-1}+m_n^2}y_{n-2}+m_{n-1}^2}\times\\ & \times & y_{n-3}+m_{n-2}^2,\\
& \cdots & ,\\
s_n^{n-2}(y_{n-1},\ldots,y_3,y_2) & = & s_n^{n-3}(y_{n-1},\ldots,y_3)-2\sqrt{s_n^{n-3}(y_{n-1},\ldots,y_3)}y_2+m_3^2.
\eeas
\begin{rem}
The iteration of square roots in particular for $s_n^{n-2}$ which is the crucial argument in $V_n^D(s_n^{n-2},m_1^2,m_2^2)$ is hopefully instructive for a future analysis of periods which emerge in the evaluation of that function \cite{allBanana}. This iteration of square roots points to the presence of a solvable Galois group  with successive quotients $\Z/2\Z$
reflecting iterated double covers in momentum space.
Thanks to Spencer Bloch for pointing this out.
\end{rem}\hfill $|$

\section{Differential Equations and related considerations}
This section collects some comments with respect to the results above with regards to:
\begin{itemize}
\item Dispersion. We want to discuss in some detail why raising powers of propagators is well-defined in dispersion integrals even if a higher power of 
a propagator consititutes a product of distributions with coinciding support.
\item Integration by parts (ibp) \cite{Tkachov}. We do not aim at constructing algorithms which can compete with the established algorithms in the standard approach
\cite{algorithms}. But at least we want to point out how ibp works in our iterated integral set-up.

\item Differential equations. Here we focus on systems of linear first order differential equations for master integrals \cite{RemODE}. We also add a few comments on higher-order differential equation for assorted master integrals which emerge as Picard--Fuchs equations \cite{Vanhove,Weinzierl,Zay,Broedel}.
\item Master integrals.
Master integrals are assumed independent by definition with regards to relations between them with coefficients which are rational functions of mass squares and kinematic invariants \cite{KalmKniehl,Panzer}. We will remind ourselves  that such a relation can still exists for their imaginary parts \cite{RemSchouten}. We trace this phenomenon back to the degree of subtraction needed in dispersion integrals to construct their real part from their imaginary parts.
Furthermore we will offer a geometric interpretation of the counting of master integrals for graphs $b_n$.      
\end{itemize}
\subsection{Dispersion and derivatives}
As we want to obtain full results from imaginary parts by dispersion we have to discuss the existence of dispersion integrals in some detail. There are subtleties when raising powers of propagators. It is sufficient to discuss the example of $b_2$.

With $\Phi_R^D(b_2)$ given consider a derivative with respect to a mass square such that a propagator is raised to second power,
\[
\Phi_R^D(b_2)_{2,1}:=\partial_{m_1^2}\Phi_R^D(b_2)(s,m_1^2,m_2^2).
\]
Similar for the imaginary part,
\[
\Im\left(\Phi_R^D(b_2)_{2,1}\right):=\partial_{m_1^2}\Im\left(\Phi_R^D(b_2)\right)(s,m_1^2,m_2^2).
\]
We have (for $D=4$ say)
\[
\Im\left(\Phi_R^4(b_2)_{2,1}\right)=\frac{s-s_0}{\pi}\int_0^{\infty}\frac{
\partial_{m_1^2}\left(\Theta(x-(m_1+m_2)^2)
V_2^4(x,m_1^2,m_2^2)\right)}{(x-s)(x-s_0)}dx.
\]
There is an issue here. It concerns the fact that to a propagator, itself a distribution,
\[
Q(r,m)=\frac{1}{r^2-m^2}\,=\mathrm{P.V.}\frac{1}{r^2-m^2}+
i\pi\delta(r^2-m^2),
\]
(using Cauchy'y principal value and the $\delta$-distribution)
we can associate a well-defined distribution by 'cutting' the propagator:
\[
\frac{1}{Q(r,m)}\to \delta_+(Q(r,m))=\Theta(r_0)\delta(r^2-m^2).
\]
The expression
\[
2\frac{\delta_+(Q(r,m))}{Q},
\]
obtained from cutting any one of the two factors in the squared propagator, 
\[
-\partial_{m^2}\frac{1}{Q(r,m)}=\frac{1}{Q^2(r,m)}
\to 
2\frac{\delta_+(Q(r,m))}{Q},
\]
is ill-defined as the numerator forces the denominator to vanish.
Hence higher powers of propagators are subtle when it comes to cuts on any one of their factors.
\begin{figure}[H]
\includegraphics[width=14cm]{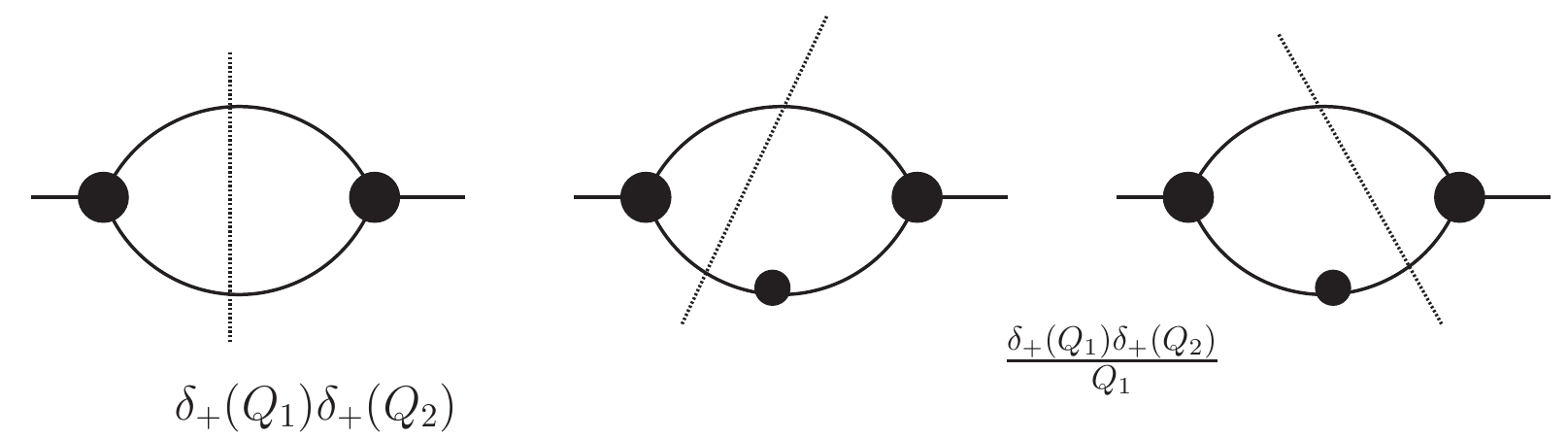}.
\caption{The doubling of propagators indicated by a dot on the edge creates a problem.}
\label{singDisp}
\end{figure}
Remarkably dispersion still works despite
the fact that derivatives like $\partial_{m_1^2}$ do just that: generating such higher powers.

We have 
\beas
\partial_{m_1^2}\Im(\Phi_R^D(b_2)) & = & \delta(s-(m_1+m_2)^2)V_2^D(s,m_1^2,m_2^2)
\left(1+\frac{m_2}{m_1}\right)\\
 & + & \Theta(s-(m_1+m_2)^2)\partial_{m_1^2}V_2^D(s,m_1^2,m_2^2),
\eeas
where
\[
\left(1+\frac{m_2}{m_1}\right)=\partial_{m_1^2} (m_1+m_2)^2.
\]
Using
\[
V_2^D(s,m_1^2,m_2^2)=\frac{\sqrt{\lambda(s,m_1^2,m_2^2)}^{D-3}}{s^{\frac{D}{2}-1}},
\]
the above is singular at $s=(m_1+m_2)^2$.
Indeed both terms on the rhs are ill-defined but their sum can be integrated in the dispersion integral
\beas
\partial_{m_1^2}\phi_R^D(b_2)  & = & \frac{(s-s_0)}{\pi}\int_0^\infty\left( \delta(x-(m_1+m_2)^2)V_2^D(x,m_1^2,m_2^2)
\left(1+\frac{m_2}{m_1}\right)\right.\\
 & + & \left. \Theta(x-(m_1+m_2)^2)\partial_{m_1^2}V_2^D(x,m_1^2,m_2^2)\right)
 \frac{1}{(x-s)(x-s_0)}dx,
\eeas
so that the singularity drops out for all $D$ by Taylor expansion of 
\[
\partial_{m_1^2}\lambda(x,m_1^2,m_2^2)=\partial_{m_1^2}\left((x-(m_1+m_2)^2)(x-(m_1-m_2)^2)\right),
\]
near the point $x=(m_1+m_2)^2$.

We are not saying that it is meaningful to replace
\[
\frac{1}{Q^2}\to \frac{\delta_+(Q)}{Q},
\]
to come to dispersion relations.

Instead  we can exchange either:\\
i) taking derivatives wrt masses on an imaginary part $\Im\left(\Phi_R^D(b_n)_\nu\right)$ first and then doing the dispersion integral,\\
or,\\
ii) doing the dispersion integral first and then taking derivatives.
\subsection{Integration-by-parts}
Integration-by-parts ($\mathrm{ibp}$) is a standard method employed in high energy physics computations.

It starts from an incarnation of Stoke's theorem in dimensional regularization
\[
0=\int d^Dk \frac{\partial}{\partial k_\mu} v_\mu F(\{k\cdot r\}),
\]
where $F$ is a scalar function of loop momentum $k$ and other momenta, and $v_\mu$ 
is a linear combination of such momenta employing a suitable definition of $D$-dimensional integration for $D\in\mathbb{C}$.

We want to discuss ibp and Stokes theorem from the viewpoint of the $y_i$-integrations in our iterated integral. 

We let $\mathbf{Int}_{b_n}$ be the integrand in Eq.(\ref{itInt}). It is made from  three factors: 
\[
\mathbf{Int}_{b_n}=\mathbf{Y}_n^{D-3}\times \mathbf{\Lambda}_n^{D-3}\times\mathbf{\sigma}_n^{1-\frac{D}{2}},
\]
with $\mathbf{Y}_n,\mathbf{\Lambda}_n,\mathbf{\sigma}_n$ defined by,
\beas
\mathbf{Y}_n^{D-3} & = & \prod_{j=2}^{n-1}\sqrt{y_j^2-m_{j+1}^2}^{D-3},\\
\mathbf{\Lambda}_n^{D-3} & = & \sqrt{\lambda(s_n^{n-2}(y_2,\ldots,y_{n-1}),m_1^2,m_2^2)}^{D-3},\\
\mathbf{\sigma}_n^{1-\frac{D}{2}} & = & \frac{1}{\left(s_n^{n-2}(y_2,\ldots,y_{n-1})\right)^{\frac{D}{2}-1}}.
\eeas
We have the following identities which allow to trade derivatives with respect to $y_j$ with derivatives with respect to $m_{j+1}^2$ or $s$
, 
\bea
\partial_{y_j}\mathbf{Y}_n & = & y_j\frac{1}{y_j^2-m_{j+1}^2}\mathbf{Y}_n  =   -2y_j\partial_{m_{j+1}^2}\mathbf{Y}_n,\label{ibpY}\\
\partial_{y_j}\mathbf{\Lambda}_n
& = & 
\frac{s_n^{n-2}-m_1^2-m_2^2}{\lambda(s_n^{n-2}(y_2,\ldots,y_{n-1}),m_1^2,m_2^2)}(\partial_{y_j} s_n^{n-2})\mathbf{\Lambda}_n\nonumber\\
      & = & \left(\partial_{s}\mathbf{\Lambda}_n\right)\left(\frac{-2\sqrt{s}}{\sqrt{s-y_{n-1}}}
      \prod_{k=1}^{n-j-1}\frac{s_n^k}{s_n^{k}-y_{n-k-1}}\right)\label{ibpL}\\
 & = & -2\sqrt{s_n^{n-j-1}}\partial_{m_{j+1}^2}\mathbf{\Lambda}_n,\nonumber\\
 \partial_{y_j} \mathbf{\sigma}_n & = &  
\partial_{y_{j}}s_{n}^{n-2}=-2\sqrt{s_n^{n-j-1}}\prod_{l=2}^{j-1}\frac{s_n^{n-l-1}-y_{l}}{s_n^{n-l-1}}\nonumber\\
 & = & -2\sqrt{s_n^{n-j-1}}
 \partial_{m_{j+1}^2}\mathbf{\sigma}_n\label{ibpS}\\
 & = &  \partial_{s}\mathbf{\sigma}_n
 \left(\frac{-2\sqrt{s}}{\sqrt{s-y_{n-1}}}
      \prod_{k=1}^{n-j-1}\frac{s_n^k}{s_n^{k}-y_{n-k-1}}\right).\nonumber  
\eea
We also note that
\be
\partial_{m_{j+1}^2}\mathbf{\Lambda}_n = (\partial_{m_{j+1}^2}s_n^{n-2})\frac{1}{m_1^2-m_2^2}\left(m_1^2\partial_{m_1^2}-m_2^2\partial_{m_2^2}\right)\mathbf{\Lambda}_n,
\ee
and
\be
\partial_{s}\mathbf{\Lambda}_n = (\partial_{s}s_n^{n-2})\frac{1}{m_1^2-m_2^2}\left(m_1^2\partial_{m_1^2}-m_2^2\partial_{m_2^2}\right)\mathbf{\Lambda}_n.
\ee

Furthermore, insertion of tensor structure given by $\nu$ following
Sec.(\ref{tensorint}) and Eqs.(\ref{tensors}-\ref{tensorsthree}) define an integrand 
$
\mathbf{Int}_{{b_n},\nu}
$. 

Now using Eq.(\ref{stokes}) we have for any such integrand,
\[
\int_{m_{j+1}}^{\mathrm{up}_n^{j+1}}\partial_{y_j}\left( \mathbf{Int}_{{b_n},\nu}\right)dy_{j}=0,\,\forall j,\, 2\leq j\leq (n-1).
\]
\begin{prop}
The above evaluates to an identity of the form,
\[
\sum_j \mathbf{Int}_{{b_n},\nu_j}=0,
\]
between tensor integrals $\mathbf{Int}_{{b_n},\nu_j}$ for some tensor structures $\nu_j$.
\end{prop} 
\begin{proof}
Derivatives with respect to $y_j$ can be traded for derivatives with respect to masses and with respect to the scale $s$
using Eqs.(\ref{ibpY},\ref{ibpS},\ref{ibpY}).  Starting with $\nu$ this creates suitable new tensorstructures $\nu_j$. Homogeneity of $\lambda$ allows to replace the $\partial_s$ derivatives by  $\mathbf{Int}_{{b_n},\tilde{\nu}_j}$ with once-more modified tensorstructures $\tilde{\nu}_j$.
\end{proof}

\subsection{Differential equations}\label{diffeq}
Functions $\Phi_R^D(G)(\{k_i\cdot k_j\},\{m_e^2\})$ for a chosen Feynman graph $G$ fulfill differential equations with respect to suitable kinematical variables \cite{RemODE}. Those variables are given by scalar products $k_i\cdot k_j$ of external momenta.
For $G=b_n$ these are differential equations in the sole scalar product $s=k_n\cdot k_n$ of external momenta.

$\Phi_R^D(b_n)(s,\{m_e^2\})$ is a solution to an inhomogeneous differential equation and the imaginary part $\Im\left(\Phi_R^D(b_n)\right)(s,\{m_e^2\})$ solves the corresponding homogeneous one.

More precisely there is a set of master integrals $\{b_n\}_M$ defined as a class of Feynman graphs such that any given graph $b_n$, giving rise to
integrals $\Phi_R^D(b_n)_\nu(s,s_0,\{m_e^2\})$, -so  with all its corresponding tensor integrals and arbitrary integer powers of propagators-, can be expressed as linear combinations of elements of $\{b_n\}_M$.

Let us consider the column vector $S_{b_n}$ formed by the elements of $\{b_n\}_M$.
One searches for a first-order system 
\[
\partial_s S_{b_n}(s)=A S_{b_n}(s)+T,
\]  
with $A=A(s,\{m_e^2\})$ a matrix of rational functions and
$T=T(\{m_e^2\})$ the inhomogenuity provided by the minors of $b_n$.
Those are $(n-1)$-loop tadpoles $t_e$ obtained from shrinking an edge $e$, $t_e=b_n/e$.

One then has
\[
\partial_s \Im(S_{b_n})(s)=A \Im(S_{b_n})(s),
\]
where  $\Im(S_{b_n})$ is formed by the imaginary parts of entries of $S_{b_n}$
and $\Im(\Phi_R^D(t_e))=0$.
 
For $b_3$ for example one has $S_{b_3}=(F_0,F_1,F_2,F_3)^T$,
with $F_0=\Phi_R^D(b_3)$, $F_i=\partial_{m_i^2}\Phi_R^D(b_3)$, $i\in \{1,2,3\}$.

The $4\times 4$ matrix $A$ and the 4-vector $T$ for that example are well-known, see \cite{Zay}.

From such a first-order system for the full set of master integrals one often derives  a higer-order differential equation for a chosen master integral. For $b_3$ or $b_4$ it is a Picard--Fuchs equation \cite{Zay}. 

For banana graphs $b_n$ it is a differential equation of order $(n-1)$:
\be\label{bananaode}
\sum_{j=0}^{n-1}\left(Q_{b_n}^{(j)}\partial_s^j\right)\Phi_R^D(b_n)(s)=T_n(s),
\ee
where $Q_{b_n}^{(j)}$ are rational functions in $s,\{m_e^2\}$ and one can always set 
$Q_{b_n}^{(n-1)}=1$.
It has been studied extensively \cite{Vanhove,Weinzierl,Zay,allBanana,Broedel}.

We want to outline how our iterated integral approach relates to such differential equations, to master integrals and to the integration-by-parts (ibp) identities which underly such structures.

Our first task is to remind ourselves how to connect the homogeneous and inhomogeneous differential equations and we turn to $b_2$ for some basic considerations.
\subsubsection{Differential equation for $b_2$}
We set $D=2$ for the moment. 
Consider the imaginary part of the bubble
\[
\Im(\Phi_R^2(b_2))(s)=\frac{1}{\sqrt{\lambda(s,m_1^2,m_2^2)}}\Theta(s-(m_1+m_2)^2).
\]
We can recover $\Phi_R^2(b_2)$ by dispersion which reads for $D=2$,
\[
\Phi_R^2(b_2)(s)=\frac{1}{\pi}\int_{(m_1+m_2)^2}^\infty\frac{\Im(\Phi_R^2(b_2))(x)}{(x-s)}dx.
\]
We now use this representation to analyse the well-known differential equation \cite{Weinzierl} for $b_2$ given in
\begin{prop}
\be\label{deb2}
\left(\lambda(s,m_1^2,m_2^2)\frac{\partial}{\partial s}+(s-m_1^2-m_2^2)\right)\Phi_R^2(b_2)(s)=\frac{1}{\pi},
\ee
and for the imaginary part
\be\label{deb2im}
\left(\lambda(s,m_1^2,m_2^2)\frac{\partial}{\partial s}+(s-m_1^2-m_2^2)\right)\Im\left({\Phi_R^2(b_2)(s)}\right)=0.
\ee 
\end{prop}
Note that Eq.(\ref{deb2im}) is the homogeneous equation associated to Eq.(\ref{deb2}) as it must be \cite{RemODE}.

The following proof aims at deriving Eq.(\ref{deb2}) from the dispersion integral.
\begin{proof}
Let us first prove Eq.(\ref{deb2im}).
\beas
\lambda(s,m_1^2,m_2^2)\frac{\partial}{\partial s}\frac{1}{\sqrt{\lambda(s,m_1^2,m_2^2)}}\Theta(s-(m_1+m_2)^2) & = & 
\frac{-(s+m_1^2+m_2^2)}{\sqrt{\lambda(s,m_1^2,m_2^2)}}\Theta(s-(m_1+m_2)^2)\\
 & + & \sqrt{\lambda(s,m_1^2,m_2^2)}\delta\left(s-(m_1+m_2)^2\right)\\
  & = & \frac{-(s+m_1^2+m_2^2)}{\sqrt{\lambda(s,m_1^2,m_2^2)}}\Theta(s-(m_1+m_2)^2)\\
   & = & -(s+m_1^2+m_2^2)\Im\left(\Phi_R^2(b_2)\right)(s),
\eeas
as desired. We use $\lambda((m_1+m_2)^2,m_1^2,m_2^2)=0$.

Now for Eq.(\ref{deb2}). Evaluating the lhs gives
\bea
LHS & = & \lambda(s,m_1^2,m_2^2)\frac{1}{\pi}\int_{(m_1+m_2)^2}^\infty\frac{1}{\sqrt{\lambda(x,m_1^2,m_2^2)}(x-s)^2}dx\label{p1}\\
& + & \frac{1}{\pi}\int_{(m_1+m_2)^2}^\infty\frac{(s-m_1^2-m_2^2)}{\sqrt{\lambda(x,m_1^2,m_2^2)}(x-s)}dx\label{p2}.
\eea
A partial integration in the first term (\ref{p1}) delivers
\beas
LHS & = & -\lambda(s,m_1^2,m_2^2)\frac{1}{2\pi}\int_{(m_1+m_2)^2}^\infty\frac{\partial_x 
\lambda(x,m_1^2,m_2^2)}{\sqrt{\lambda(x,m_1^2,m_2^2)}^3(x-s)}dx\\
 & & +\frac{1}{\pi}\int_{(m_1+m_2)^2}^\infty\frac{(s-m_1^2-m_2^2)}{\sqrt{\lambda(x,m_1^2,m_2^2)}(x-s)}dx\\
  & & -
 \lambda(s,m_1^2,m_2^2)\left[ \frac{1}{\pi}\frac{1}{\sqrt{\lambda(x,m_1^2,m_2^2)}(x-s)}\right]^\infty_{(m_1+m_2)^2}.
\eeas

We have
\be\label{qone}
\partial_x 
\lambda(x,m_1^2,m_2^2)=2(x-m_1^2-m_2^2)=:v_1(x),\, v_1(x)-v_1(s)=2(x-s),
\ee 
and
\be\label{qzero}
\lambda(s,m_1^2,m_2^2)-\lambda(x,m_1^2,m_2^2)=(s-x)((s+x)-2(m_1^2+m_2^2))=:w(x,s)(s-x).
\ee 
Using this the lhs of Eq.(\ref{deb2}) reduces to a couple of boundary terms.
We collect
\beas
 & + & \frac{1}{\pi}\left[ \frac{x}{\sqrt{\lambda_x}}\right]^\infty_{(m_1+m_2)^2}\\
 & + & \frac{s-2(m_1^2+m_2^2)}{\pi}\left[ \frac{1}{\sqrt{\lambda_x}}\right]^\infty_{(m_1+m_2)^2}\\
 &  + &
  \left[ \frac{1}{\pi}\frac{(s-x)w(s,x)+\lambda_x}{\sqrt{\lambda_x}(x-s)}\right]^\infty_{(m_1+m_2)^2}\\
   & = & \frac{1}{\pi},
\eeas
as desired.

Indeed using that $w(s,x)=s+x-2(m_1^2+m_2^2)$ we see that the term $\sim w$ in the third line cancels the first and second line.
The remaining term is
\[
  \left[ \frac{1}{\pi}\frac{\lambda_x}{\sqrt{\lambda_x}(x-s)}\right]^\infty_{(m_1+m_2)^2}=\frac{1}{\pi},
\]
as $\sqrt{\lambda((m_1+m_2)^2,m_1^2,m_2^2)}=0$ and 
$\lim_{x\to\infty}\sqrt{\lambda(x,m_1^2,m_2^2)}=x$.
\end{proof}

\begin{rem}
So for $b_2$ we have by Eqs.(\ref{qzero},\ref{qone})
\[
Q_0(x)=\frac{2(s-m_1^2-m_2^2)}{\lambda(s,m_1^2,m_2^2)} {\text{ and }} Q_1(x)=1.
\]
This is a trivial incarnation of Eq.(\ref{bananaode}).
As $(Q_0(x)-Q_0(s))\sim (x-s)$ we cancel the denominator $1/(x-s)$ in the dispersion integral and we are left with boundary terms which constitute the inhomogeneous terms.
\end{rem}
\begin{rem}
The non-rational part 
$\Phi_R^D(b_2)_\mathbf{Transc}$ of $\Phi_R^D(b_2)$ is divisible by $V_2^D$ and gives a pure function
in the parlance of \cite{BroedeletalEll}. Indeed one wishes to dentify such pure functions in the non-rational parts of $\Phi_R^D(b_n)(s,s_0)$.

For example for $D=4$ (ignoring terms in 
$\Phi_R^4(b_2)(s)$ which are rational in $s$)
\[
\Phi_R^4(b_2)(s)_\mathbf{Transc}/V_2^4(s)=\ln\frac{m_1^2+m_2^2-s-\sqrt{\lambda(s,m_1^2,m_2^2)}}{m_1^2+m_2^2-s+\sqrt{\lambda(s,m_1^2,m_2^2)}}.
\] 
 This follows also for all $b_n$, $n>2$, as long as the inhomogenuity $T_n(s)$ fulfils 
\[
\Im(T_n(s))=0,
\]
which is certainly true for the case $b_2$ with $T_2(s)=1/\pi$.
Indeed, for $f(s)$ a solution of the homogeneous
\[
\left(\sum_{j=0}^{n-1} Q_j(s)\partial^j_s\right) f(s)=0,
\]
the inhomogeneous Picard--Fuchs equation
\[
\left(\sum_{j=0}^{n-1} Q_j(s)\partial^j_s\right) g(s)=T_n(s),
\]
can be solved by setting $g(s)=f(s)h(s)$.
Using Leibniz' rule this determines $h(s)$ as a solution of an equation
\[
\sum_{k=1}^{n-1} h^{(k)}(s)\left(\sum_{j=k}^{n-1} \left(j \atop k \right) a_j(s) f^{(j-k)}(s)\right)=u(s),
\]
with $f^{(j-k)}(s)=\partial_s^{j-k}f(s)$ and similarly for $h^{(k)}(s)$. 
Note $f^{(j-k)}(s)$ are given by solving the homogeneous equation.
Hence $g(s)$ indeed factorizes as desired.\footnote{The argument can be extended by replacing the requirement $\Im(T_n(s))=0$ by $\mathbf{Var}_x(T_n(s))=0$ where 
$\mathbf{Var}_x$ is the variation around a given threshold divisor $x$.
For banana graphs $b_n$ we only have to consider $x=s_{\mathbf{normal}}$.}

This relates to co-actions and cointeracting bialgebras \cite{KarenDirk,DirkCoact} and will be discussed elsewhere.
\end{rem}
\subsubsection{Systems of linear differential equations for $b_n$} 
To find differential equations for the iterated $y_j$-integrations of
Eq.(\ref{itInt}) we first systematically shift all $y_j$-derivatives acting on
$\sqrt{y_j^2-m_{j+1}^2}$ to act on $V_2^D(s_n^2,m_1^2,m_2^2)$ using partial integration. We can ignore boundary terms by Thm.(\ref{monodromyThm} iii)).
We use
\beas
\left(\partial_{m_j^2}\frac{1}{\sqrt{y_{j-1}-m_j^2}}\right)F & = &
\frac{1}{2\sqrt{y_{j-1}-m_j^2}^3}F\\
 & = &  \left(-\frac{y_{j-1}^2-m_j^2}{2m_j^2\sqrt{y_{j-1}-m_j^2}^3}
 +\frac{y_{j-1}^2}{2m_j^2\sqrt{y_{j-1}-m_j^2}^3}\right)F\\
  & = & \left(-\frac{1}{2m_j^2\sqrt{y_{j-1}-m_j^2}}
 -y\left(\partial_{y_{j-1}}\frac{1}{2m_j^2\sqrt{y_{j-1}-m_j^2}}\right)\right)F\\
  & = & -\frac{1}{2m_j^2\sqrt{y_{j-1}-m_j^2}}F
 +\frac{1}{2m_j^2\sqrt{y_{j-1}-m_j^2}}\left(\partial_{y_{j-1}}y_{j-1}F\right)\\
 & = & +\frac{1}{2m_j^2\sqrt{y_{j-1}-m_j^2}}y_{j-1}\left(\partial_{y_{j-1}} F\right)\\
  & = & +\frac{1}{m_j^2\sqrt{y_{j-1}-m_j^2}}y_{j-1}\left(\sqrt{s_n^{j-1}}\partial_{m^2_{j}} F\right).
\eeas
We could trade a derivative wrt $y_{j-1}$ for a derivative wrt $m_j^2$ thanks to Thm.(\ref{monodromyThm} iv)). This holds under the proviso that all masses are different. Else we use the penultimate line as our result:
\[
\left(\partial_{m_j^2}\frac{1}{\sqrt{y_{j-1}-m_j^2}}\right)F=
+\frac{1}{2m_j^2\sqrt{y_{j-1}-m_j^2}}y_{j-1}\left(\partial_{y_{j-1}} F\right).
\] 
We can iterate this and shift higher than first  derivatives
\[
\left(\partial_{m_j^2}^k\frac{1}{\sqrt{y_{j-1}-m_j^2}}\right)F
\] 
to derivatives on $F$.

We note that from the definition of $\lambda(s_n^{n-2},m_2^2,m_1^2)$ we have
\[
\lambda(s_n^{n-2},m_2^2,m_1^2)=s_n^{n-2}(s_n^{n-2}-2(m_1^2+m_2^2))+(m_1^2-m_2^2)^2.
\]
By Euler ($\lambda$ is homogeneous of degree two), 
\beas
2\lambda(s_n^{n-2},m_2^2,m_1^2)  & = & \partial_{s_n^{n-2}}\lambda(s_n^{n-2},m_2^2,m_1^2)
+\partial_{m_1^2}\lambda(s_n^{n-2},m_2^2,m_1^2)+\partial_{m_2^2}\lambda(s_n^{n-2},m_2^2,m_1^2).
\eeas

Also,
\beas
\partial_{m_1^2}\lambda(s_n^{n-2},m_2^2,m_1^2) & = & 2(m_1^2-m_2^2-s_n^{n-2}),\\
\partial_{m_2^2}\lambda(s_n^{n-2},m_2^2,m_1^2) & = & 2(m_2^2-m_1^2-s_n^{n-2}),\\
\partial_{m_j^2}\lambda(s_n^{n-2},m_2^2,m_1^2) & = & 2(s_n^{n-2}-m_1^2-m_2)\partial_{m_j^2}s_n^{n-2},\,\forall\, 3\leq j\leq n,\\
\partial_{s}\lambda(s_n^{n-2},m_2^2,m_1^2) & = & 2(s_n^{n-2}-m_1^2-m_2)
\partial_s s_n^{n-2}.
\eeas
With this Thm.(\ref{monodromyThm}) allows to derive differential equations.

Let us rederive for example the differential equation for the three-edge banana.
Let us define
\beas
F_0 & = & \Phi_R(b_3),\\
F_1 & = & \partial_{m_1^2} F_0,\\
F_2 & = & \partial_{m_2^2} F_0,\\
F_3 & = & \partial_{m_3^2} F_0,\\
F_s & = & \partial_s F_0.
\eeas
Then we have
\be\label{diffF0}
(D-3)F_0+\sum_{j=1}^3 m_j^2 F_j=s\partial_s F_0,
\ee
and similarly
\be\label{diffFj}
\left((D-4)+\sum_{i=1}^3m_i^2\partial_{m_i^2}\right) F_j=s\partial_s F_j,\,j\in\{1,2,3\}.
\ee

The integrands $I_i$ for $(D-3)F_0$,$m_1^2F_1$,$m_2^2F_2$,$m_3^2F_3$, and $sF_s$ can be written as
\[
I_{i}=\frac{\mathbf{num}_i(y_2)}{s^{\frac{D}{2}}}\sqrt{y_2^2-m_3^2}^{D-5}
\sqrt{\lambda}^{D-5}(s_3^1,m_2^2,m_1^2)
\]
with suitable polynomials $\mathbf{num}_i$ in $y_2$.
Eq.(\ref{diffF0}) follows immediately as the corresponding numerators 
${\mathbf{num}_i(y_2)}$ add to zero.

The equations Eqs.(\ref{diffFj}) for $F_1,F_2,F_3$ can be proven in precisely the same manner and many more differential equations follow from using the ibp identities Eqs.(\ref{ibpY}-\ref{ibpS}).

Furthermore $F_0,F_1,F_2,F_3$ provide master integrals for the 
Feynman integrals $\Phi_R^D(b_3)_\nu$ \cite{Zay}.
\begin{rem}
Note that we can infer the independence of $F_0,F_1,F_2,F_3$ from the fact that the corresponding polynomials  are different, in fact of different degree in $y_2$.

We could also use different integral representations for $F_1,F_2,F_3$ by setting
\beas
F_3=\partial_{m_3^2}\,{\text{rhs\,of\,Eq.(\ref{bthreeo})}},\\
F_2=\partial_{m_2^2}\,{\text{rhs\,of\,Eq.(\ref{bthreetw})}},\\
F_1=\partial_{m_1^2}\,{\text{rhs\,of\,Eq.(\ref{bthreeth})}}.
\eeas
and conclude from there.
\end{rem}\hfill$|$
\subsection{Master integrals}\label{masterint}
We want to comment on two facts:\\
i)  a geometric interpretation of the known formula for the counting of master integrals for $b_n$,\\
ii) that the independence of elements $x$ of a set $S_ {b_n}$ of master integrals does not imply the independence of elements of $\Im(x),\, x\in\left(S_ {b_n}\right)$.   
\subsubsection{A geometric interpretation: Powercounting} Let us start with a geometric interpretation. We collect a well-known proposition \cite{KalmKniehl,Panzer}.
\begin{prop}
The number of master integrals for the $n$-edge banana with different masses is $2^n-1$.
\end{prop}
Let us pause. For $b_3$, we have four master integrals, $F_0$, and three possibilities to put a dot on an internal edge. Furthermore, we can shrink any of the three internal edges, giving us three two-petal roses as minors.
This makes $7=2^3-1$ master integrals amounting to the fact that all tensor integrals $\Phi_R(b_n)\nu$ can be expressed as a linear combination of those seven, with coefficients which are rational functions in the mass-squares and in $s$.   

Similarly for $b_4$ we have $\Phi_R^D(b_4)$ itself, four integrals
$\partial_{m_i^2}\Phi_R^D(b_4)$ and six $\partial_{m_j^2}\partial_{m_i^2}\Phi_R^D(b_4)$, $i\not= j$. There are four minors as well, so that we get the desired $15=2^4-1$ master  integrals.

For arbitrary $n$ there are indeed $\left( n\atop j\right)$ possibilities to put one dot on $j$ edges,
and 
\[
\sum_{j=0}^{n-2}\left( n\atop j\right)=2^{n}-n-1,
\] 
possibilities to put a single dot on up to $n-2$ edges. Furthermore we have $n$ minors from shrinking one of the $n$ edges, so we get $2^n-1$ master integrals.

Furthermore it is obvious from the structure 
of the iterated integral in Eq.(\ref{itInt}) that the two edges forming the innermost $b_2$ do not need a dot. Indeed the corresponding  loop integral 
in $k_1$ is fixed by two $\delta_+$ functions.  Integration by parts then ensures that we do not need more than one dot per edge at most. 

\begin{rem}
One can analyse this from the viewpoint of powercounting. Let us choose $D=4$ so that $b_2$ is log-divergent.
Let us note that for $D=4$ 
\be\label{masterzero} 
4(n-1)-2\overbrace{(2n-2)}^{\# E}=0,
\ee
where $\# E$ is the number of edges of a banana graph $b_n$ which has  $(n-2)$ edges
with a single dot each. Eq.(\ref{masterzero}) says that $b_n$ furnished with the maximum of $n-2$ dots gives an overall logarithmic singular integral for any $n$.
 
A lesser number of dots gives a higher degree of divergence and hence higher subtractions in the dispersion integrals. Conceptually, higher degrees of divergence are probing higher coefficients in the Taylor expansion in $s$ which provide the needed master integrals.
We see below how this interferes with counting master integrals but first our geometric interpretation as given in Fig.(\ref{theta}).
\hfill $|$
\end{rem}
\subsubsection{$b_3$ and its cell}
The parametric representation of $b_3$ as given in App.(\ref{bthreeparamet})
provides insight into the structure of its Feynman integral and the related master integrals.
\begin{rem}
Let us note that any graph $b_n$ has a spanning tree which consists of just one of its internal edges. Hence any associated spanning tree has length one.
As $b_n$ has $n$ internal edges its associated cell $C(b_n)$  (in the sense of {\em Outer Space} \cite{CullVogt})
is a $(n-1)$-dimensional simplex $C_n$
\[
C(b_n)=C_n.
\] 
The graph $b_n$ has internal edges $e_i$. To each such edge we assign a length $A_i$, $0\leq A_i\leq \infty$ which we regard as a coordinate in the projective space $\mathbb{P}_{b_n}:=\mathbb{P}^{n-1}(\mathbb{R}_+)$. 

Shrinking one edge $e_i$ to length $A_i=0$
gives the graph $b_n/e_i$ which is associated to the codimension-one boundary
determined by $A_i=0$. It is a $(n-2)$-dimensional simplex $C_{n-1}$.

Note $b_n/e_i$ is a rose with $(n-1)$ petals. Each petal corresponds to a tadpole integral for a propagator with mass $m_j^2$, $j\not= i$.

Different points of $C(b_n)$ correspond to different points 
\[
\mathbb{P}_{b_n}\ni p:\,(A_1:A_2:\cdots:A_n).
\]  
We can identify $n!$ sectors $\mathbf{\sigma}:\,A_{\sigma(1)}>A_{\sigma(2)}>\cdots>A_{\sigma(n)}$
for any permutation $\sigma\in S_n$ with associated sector $\mathbf{\sigma}$.
\be
\Phi_R^D(b_n)(s,s_0)=\int_{\mathbb{P}_{b_n}(\mathbb{R}_+)}\mathrm{Int}_{b_n}(s,s_0;p)=\sum_{\sigma\in S_n}\int_{\mathbf{\sigma}}T^{(\rho^n_D)}\left[\mathrm{Int}_{b_n}(s,s_0;p)\right],
\ee
with
\[
\mathrm{Int}_{b_n}(s,s_0;p)=\frac{\ln\frac{\Phi(b_n)(s)(p)}{\Phi(b_n)(s_0)(p)}}{\psi_{b_n}^{\frac{D}{2}}(p)}\Omega_{b_n}.
\]
$T^{(\rho^n_D)}$ is a suitable Taylor operator with subtractions at $s=s_0$ ensuring overall convergence
and  $\rho^n_D$ the UV degree of divergence.
Here,
\[
\Phi(b_n)(s)(p)=\left(\prod_{j=1}^n A_j\right)\underbrace{\left(s-\left(\sum_{i=1}^n A_im_i^2\right)\left(\sum_{k=1}^n \frac{1}{A_k}\right)\right)}_{TP(b_n)},
\]
and
\[
\left(\prod_{j=1}^n A_j\right)\left(\sum_{k=1}^n \frac{1}{A_k}\right).
\]
Each sector allows for a rescaling acording to the order of edge variables such that 
the singularity is an isolated pole. 

Here $TP(b_n)$ is the toric polynomial of $b_n$
as discussed in \cite{allBanana,VanhRev} and prominent in the GKZ approach used there. 

Such approaches with their emphasis on hypergeometrics and the r\^ole of confluence have a precursor in the study of Dirichlet measures \cite{Carlsen}.
The latter have proved their relevance for Feynman diagram analysis early on \cite{DirkCarlsen}.

The  spine of $C(b_n)$ is a $n$-star, with the vertex in the barycenter 
and $n$ rays from the barycenter $bc$ of $C(b_n)$ to the midpoints of the $n$ codimension-one cells $C_{n-1}$ which are $(n-2)$-simplices.

These rays provide $n$  corresponding cubical chain complices $\mathrm{cc}(i)$  each provided by  single intervals $[0,1]$. 

For the two endpoints $0$ and $1$ of each $\mathrm{cc}(i)$, we assign:\\
i) to $1$, -the barycenter $bc$ common to all $\mathrm{cc}(i)$ we assign $b_n$ with internal edges removed, hence evaluated on-shell. This corresponds to $\Im\left(\Phi_R^D(b_n)\right)$.\\
ii)  to $0$, we assign the graph $b_n/e_i$ (a rose with $n-1$ petals) with petals of equal size  -hence a tadpole $\Phi_R^D(b_n/e_i)$
with $A_jm_j=A_km_k$, $j,k\not= i$. See Fig.(\ref{theta}). 
\end{rem}\hfill$|$
\begin{figure}[h]
\includegraphics[width=14cm]{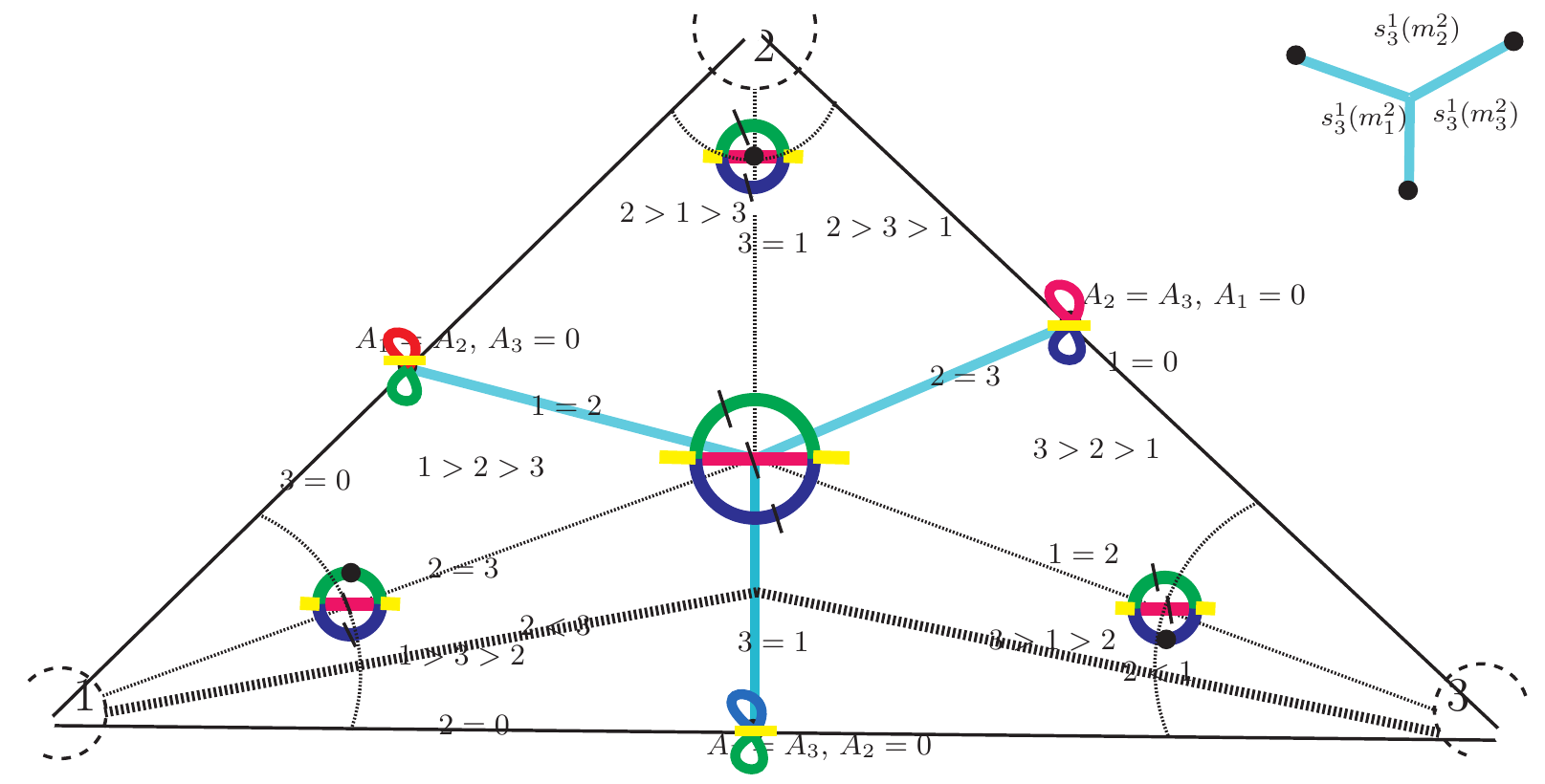}.
\caption{The graph $b_3$  and its triangular cell $C_3$. The codimension-one boundaries (sides) are given by the condition $A_i=0$, indicated in the figure by $i=0$, $i\in\{1,2,3\}$. The graph $b_3$ with two yellow leaves as external edges is put in the barycenter.
All its edges are put on-shell.  The cell decomposes into six sectors $m_iA_i>m_jA_j>m_kA_k$ as indicated by $i>j>k$. The lines $m_iA_i=m_jA_j$ (indicated by $i=j$) start at the midpoint 
$\mathrm{mid}_{i,j}:\, A_k=0,\,A_im_i=A_jm_j$ of the co-dimension one boundary $A_k=0$ and pass through the barycenter $\mathrm{bc}:\,m_1A_1=m_2A_2=m_3A_3$ towards the corner $c_k:\,A_i=A_j=0$, labeled $k$. Such corners are removed.
For these three lines the three intervals $[\mathrm{mid}_{i,j},\mathrm{bc}]$ from the midpoints of the sides to the barycentre of the cell form the spine. 
It indicated in turquoise.  The bold hashed line indicated by $2<3$ 
(so $m_2A_2<m_3A_3$) on the left and $2<1$ (so $m_2A_2<m_1A_1$) on the right is an example of a fibre over one (the vertical) part (on the $1=3$-line) of the spine (the turquoise line from $m_1A_1=m_3A_3,A_2=0$ to the barycentre). On the left, along the fibre the ratio $A_2/A_3<m_3/m_2$ is a constant, 
on the right similarly.
Finally, to the two yellow leaves we assign incoming four-momenta $k_3,-k_3$ with $k_3^2=s$. The spine partitions the cell $C_3$ into three 2-cubes, boxes $\Box(j)$ with four corners for any $\Box(j)$:
$\mathrm{mid}_{i,j},\mathrm{bc},\mathrm{mid}_{j,k}, c_j$. For each such box $\Box(j)$ there is a diaginal $d_j$.
It is a line from a corner to the barycenter: $\mathrm{d}_j:\,]c_j,\mathrm{bc}]$ for which we have $m_iA_i=m_kA_k$. We assign to this diagonal $\mathrm{d}_j$  a graph for which edges $e_i,e_k$ are on-shell and edge $e_j$ carries a dot. Along the diagonal $\mathrm{d}_j$ we have $A_jm_j>(A_im_i=A_km_k)$.}
\label{theta}
\end{figure}

Fig.(\ref{theta}) gives the graph $b_3$ and the associated cell, a 2-simplex $C_3$.
It is a triangle with corners $c_1,c_2,c_3$. Points of the cell are the interior points of $C_3$ and 
furthermore the points in the three codimension-one boundaries $C_2(i)$, the sides of the triangle.

The corners $c_i$  are removed and do not belong to the cell. Points of the cell parameterize the edge lengths $A_i$  of the internal edges of $b_3$
as parameters in the parametric integrand, see Eq.(\ref{bthreepar}). 

The boundaries
are given by $C_2(i): A_i=0$ and correspond to tadpole integrals for tadpoles $t_2(i)=b_3/e_i$ for which edge $e_i$ has length zero. 

Corners $c_k:\,A_i=A_j=0,\,i\not= j $ correspond to $b_3/e_i/e_j$ which 
is degenerate as it shrinks a loop.

  Colours green, red, blue indicate three different masses. It is understood that a momentum $k_3$ flows through any edge $e_i$ which is chosen to serve as a spanning tree for $b_3$.

The three edges of the graph give rise to $3!$ orderings of the edge lengths as indicated in the figure. We will split the parametric integral accordingly. See App.(\ref{bthreeparamet}) for computational details.

To a $(i=j)$-diagonal of a box $\Box(k)$  we associate a $b_3$ evaluated  with edges $e_i,e_j$ on-shell and edge $e_k$ dotted, so it corresponds to $\partial_{m_k^2}\Im\left(\Phi_R^D(b_3)\right)$. 

In the figure there is also an arc given which is a fibre which has the diagonal $d_j$  as the base. Integrating that fibre corresponds to integrating the $b_2$ subgraph on edges $e_i,e_j$. Points $(A_i:A_j:A_k)$ on a diagonal $d_k$ fulfil 
\[
A_km_k>x,\, x:=A_im_i=A_jm_j.
\]

To the barycentre $A_im_i=A_jm_j$ we associate $b_3$ with all three edges on-shell, a Cutkosky cut providing $\Im\left(\Phi_R^D(b_3)\right)$. To the midpoints
$A_i=A_j,A_k=0$ of the edges $A_i=0$ ($e_i=0$ in
the figure) we assign tadpole integrals. All in all we identified all seven master integrals in the figure. Note that the cell decomposition in Fig.(\ref{theta}) reflects the structure of the Newton polyhedron associated to $TP(b_3)$ \cite{VanhRev}.

Note that the requirement $A_im_i=A_jm_j$ is the locus for the Landau singularity of the associated $b_2(e_i,e_j)$ and similarly for $A_1m_1=A_2m_2=A_3m_3$ and $b_3$.
\begin{rem}
Note that the diagonals $d_j$ can be obtained by reflecting a leg of the spine at the barycenter. The three legs and the three diagonals form the six boundaries between the sectors $A_i>A_j>A_k$. 

A similar analysis holds for any $b_n$. 
For example for $b_4$ the cell is a tetrahedron with four corners $c_i$, $i\in\{1,2,3,4\}$.  The spine is a 4-star with four lines (rays) from the barycenter $bc:\,m_1A_1=m_2A_2=m_3A_3=m_4A_4$ to the midpoints of the four sides of the tetrahedraon (triangles). Reflecting these lines at the barycenter gives four diagonals $d_j:\,[bc,c_j]$
from $bc$ to one of the four corners $c_i$. 

To $bc$ we associate $\Im\left(\Phi_R^D(b_4)\right)$. To the diagonals $d_j$ we assign $\partial_{m_j^2}\Im\left(\Phi_R^D(b_4)\right)$ with the edges $e_i,i\not= j$, on-shell. There are six triangles with sides $d_i,d_j,]c_i,c_j[$. To those we assign  
$\partial_{m_i}^2\partial_{m_j^2}\Im\left(\Phi_R^D(b_4)\right)$
with the edges $e_k,k\not= i,j$, on-shell.
See Fig.(\ref{bfourtetra}).
\end{rem}\hfill$|$
\begin{figure}[H]
\includegraphics[width=14cm]{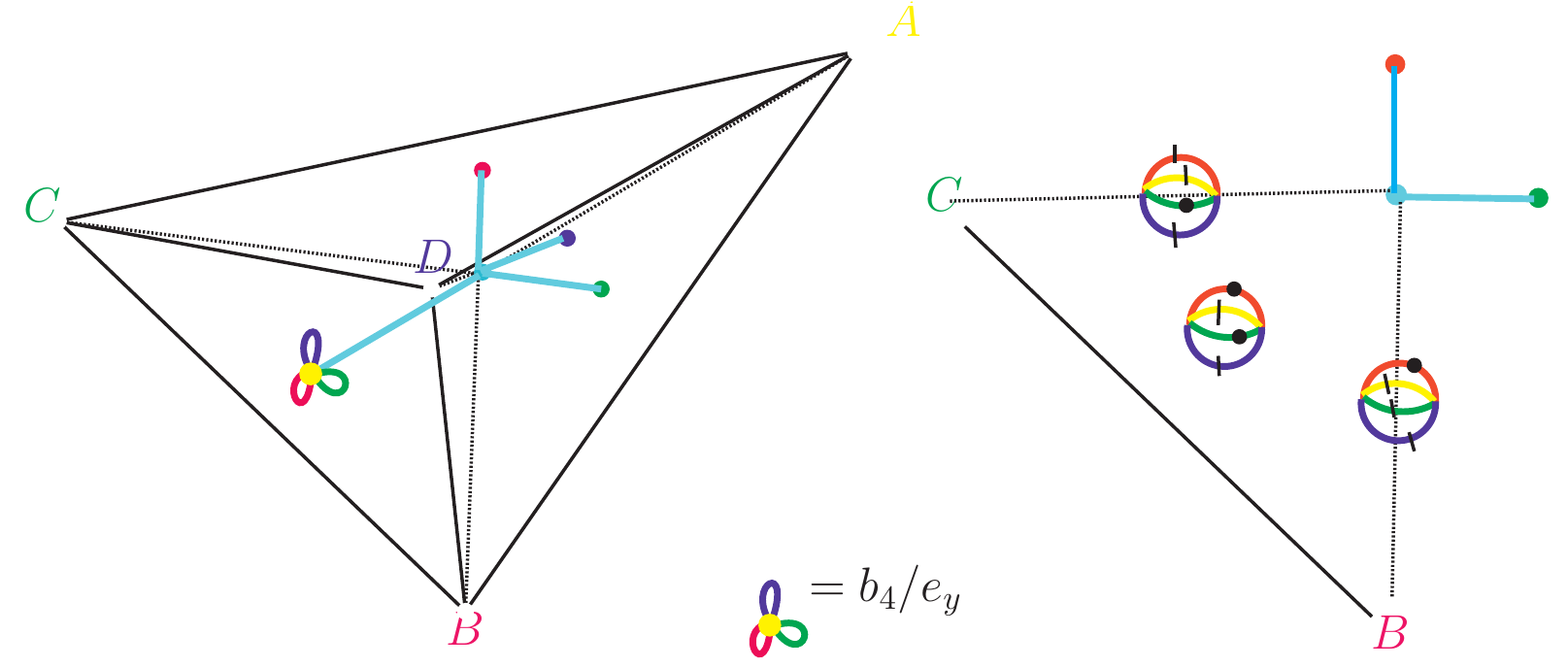}.
\caption{The cell $C(b_4)=C_4$ on the left. On the right we see two diagonals 
$d_C,d_B$ and their associated graphs which have one dotted edge. Points of the triangle $bc,B,C$ are the open convex hull of $d_C,d_B$ which we denote as the span of the diagonals $d_C,d_B$. To them a graph with two dotted  edges is assigned. On the codimension-one triangles spanned by three corners we indicate
the barycentre by a coloured dot. For example to the triangle $BCD$ we have the yellow dot and the graph $b_4/e_y$ assigned to it where the yellow edge shrinks to length zero.}
\label{bfourtetra}
\end{figure}
Continuing we get the expected tally: for $b_n$, we have $\binom{n}{0}=1$ graph for the barycenter,
$\binom{n}{1}=n$ graphs for the diagonals, $\binom{n}{m},\ m\leq (n-2)$ graphs for the span of $m$ diagonals, and $\binom{n}{n-1}=n$ tadpoles. 
It is rather charming to see how mathematics inspired by the works of Karen Vogtmann and collaborators \cite{CullVogt} illuminates results discussed recently in terms of intersection theory \cite{intersect}.

\subsubsection{Real and imaginary independence and powercounting}
Next we want to compare real and imaginary parts to check that the independence of elements of $S_ {b_n}$ does not necessarily imply the independence of elements of $\Im\left(S_ {b_n}\right)$.   We demonstrate this well known fact \cite{RemSchouten} for $b_3$. Independence is indeed a question of the values of $D$.

For $b_3$ and $D=2$ we need no subtraction in the dispersion integral for $F_0=\Phi_R^2(b_3)$,
\[
\Phi_R^2(b_3)(s)=\frac{1}{\pi}\int_{(m_1+m_2+m_3)^2}^\infty\frac{V_3^D(x,m_1^2,m_2^2,m_3^2)}{(x-s)}dx, 
\]
and  for $F_i=\partial_{m_i^2} F_0$ again an unsubtracted dispersion integral suffices
\[
F_i(s)=\frac{1}{\pi}\int_{(m_1+m_2+m_3)^2}^\infty\frac{\partial_{m_i^2}V_3^D(x,m_1^2,m_2^2,m_3^2)}{(x-s)}dx. 
\]
The four integrands $I_i$ (for the $y_2$-integration)  of $\Im(F_i)$, $i\in\{0,1,2,3\}$ can be expressed over a common denominator with numerators 
 $\mathbf{num}_i(y_2)$ and for $D=2$ 
(the $(s_n^{n-2})^{\frac{D}{2}-1}=1$ is absent) there is indeed a relation between the four numerators.
\be\label{dependence}
\mathbf{num_3}(y_2)=c_{0}^3 \mathbf{num}_0(y_2)+c_{1}^3 \mathbf{num}_1(y_2)+c_{2}^3 \mathbf{num_2}(y_2),
\ee
where $c_{i}^3$ are rational functions of $s,m_1^2,m_2^2,m_3^2$ independent of $y_2$. 

For $D=2$ a second relation follows from the fact that the integrand involves the square root of a quartic polynomial (\cite{RemSchouten}, App.(D)),
\[
\frac{1}{\sqrt{y_2^2-m_3^2}}V_3^2(y_2)=\frac{1}{\sqrt{s}\sqrt{(y_2-m_3)(y_2+m_3)(y_2-y_+)(y_2-y_-)}},
\]
where we set for the quadratic polynomial
$\lambda(s_3^1(y_2),m_1^2,m_2^2)$, 
\[
\lambda(s_3^1(y_2),m_1^2,m_2^2)=:s(y_2-y_+)(y_2-y_-),
\]
which defines $y_\pm$. See Sec.(\ref{bthreeelliptic}).

Investigating
\[
J_n=\int_{m_3}^{\mathbf{up}_3^0} \frac{y_2^n}{\sqrt{s}\sqrt{(y_2-m_3)(y_2+m_3)(y_2-y_+)(y_2-y_-)}}dy_2,
\]
as in \cite{RemSchouten} delivers a further relation between the $F_i$ and we are hence left with only two independent master integrals for the imaginary parts of $b_3$ in $D=2$.

For $b_3$ and $D=4$ on the other hand we need a double subtraction in the dispersion integral for $F_0=\Phi_R^4(b_3)$,
\[
\Phi_R^4(b_3)(s,s_0)=\frac{(s-s_0)^2}{\pi}\int_{(m_1+m_2+m_3)^2}^\infty\frac{V_3^D(x,m_1^2,m_2^2,m_3^2)}{(x-s)(x-s_0)^2}dx, 
\]
whilst for $F_i=\partial_{m_i^2} F_0$ a once subtracted dispersion integral suffices,
\[
F_i(s)=\frac{(s-s_0)}{\pi}\int_{(m_1+m_2+m_3)^2}^\infty\frac{\partial_{m_i^2}V_3^D(x,m_1^2,m_2^2,m_3^2)}{(x-s)(x-s_0)}dx. 
\]
The four integrands $I_i$ (for the $y_2$-integration)  of $\Im(F_i)$, $i\in\{0,1,2,3\}$ have to be expressed over a different common denominator $D=4$,
in particular having an extra factor $s_3^1$.
There is no relation between them.

This reflects the fact that the $F_0$ dispersion 
\[
\Phi_R^4(b_3)(s,s_0)=\frac{(s-s_0)}{\pi}\int_{(m_1+m_2+m_3)^2}^\infty\left(\frac{V_3^D(x,m_1^2,m_2^2,m_3^2)}{(x-s)(x-s_0)}
-\frac{V_3^D(x,m_1^2,m_2^2,m_3^2)}{(x-s_0)^2}\right)dx, 
\]
subsumes the Taylor expansion $s$ near $s_0$ to second order.

In contrast the $F_i$, $i\in\{1,2,3\}$, 
\[
\partial_{m_i^2}\Phi_R^4(b_3)(s,s_0)=\partial_{m_i^2}\frac{1}{\pi}\int_{(m_1+m_2+m_3)^2}^\infty\left(\frac{V_3^D(x,m_1^2,m_2^2,m_3^2)}{(x-s)}
-\frac{V_3^D(x,m_1^2,m_2^2,m_3^2)}{(x-s_0)}\right)dx, 
\]
subsume the Taylor expansion in $s$ near $s_0$ to first order.

This is in agreement with the powercounting in Eq.(\ref{masterzero})
and forces the relation between the four $F_i$ to be $\sim s\partial_s F_0$,
see Eq.(\ref{diffF0}). The relation Eq.(\ref{dependence}) is spoiled by the extra 
coefficient in the Taylor expansion of $\Phi_R^4(b_3)(s,s_0)$.

We are left with four, not two, master integrals. Indeed, starting with a dotted log-divergent banana integral reducing the number of dots demands more subtractions in the dispersion integral. Any relation between imaginary parts with different numbers of dots is spoiled by the difference in degree needed for the subtractions in the dispersion integral.
\begin{appendix}
\section{Feynman rules for banana graphs}\label{AppFeyn}
\label{sec:1}
\vspace{1mm}\noindent
Having introduced the graphs $b_n$ as our subject of interest we define Feynman rules for their evaluation. We follow the momentum routing as indicated in 
Fig.(\ref{bananas}).

The graph $b_n$ gives rise to an integrand $I_{b_n}$ (setting $k_0=(0,\vec{0})^T$, where the $D$-vector 
$k_0$ is set to the zero-vector $(0,\vec{0})^T\in \mathbb{M}^D$):
\[
I_{b_n}=\omega^D_{(n-1)}\prod_{j=0}^{n-1}\frac{1}{(k_{j+1}-k_j)^2-m_{j+1}^2},
\]  
and we set $Q_{j+1}=(k_{j+1}-k_{j})^2-m_{j+1}^2$, $0\leq j\leq (n-1)$
for the $n$ quadrics $Q_{j+1}$, $j=0,\ldots,n-1$.
Here
\[
\omega^D_{(n-1)}:=d^Dk_1\cdots d^Dk_{n-1}
\] 
is a $D\times(n-1)$-form in a $(n-1)$-fold product $\mathbb{M}_n$ of $D$-dimensional Minkowski spaces 
\[
\mathbb{M}_n:= \left(\mathbb{M}^D\right)^{\times (n-1)}.
\]

The function $\Phi_R^D(b_n)(s)$ is multi-valued as a function of $s:=k_n^2$.
It has an imginary part given by a cut which amounts to replacing for each quadric
\[
\frac{1}{Q_{j+1}}\to \delta_+((k_{j+1}-k_j)^2-m_{j+1}^2),
\]
in the integrand $I_{b_n}$. This is Cutkosky's theorem \cite{Cutkosky} applied to $b_n$.
The distribution $\delta_+$ acts as
\[
\delta_+((k_{j+1}-k_j)^2-m_{j+1}^2)=\Theta(k_{j+1;0}-k_{j;0})\delta((k_{j+1}-k_j)^2-m_{j+1}^2),
\]
using the Heavyside distribution $\Theta$ and Dirac $\delta$-distribution.

The integrand for the cut banana is correspondingly
\be\label{icut}
I_{\mathrm{cut}}(b_n)=
\omega^D_{(n-1)}\prod_{j=0}^{n-1}\delta_+((k_{j+1}-k_j)^2-m_{j+1}^2).
\ee
We take the external momentum $k_n$ to be timelike so that we can choose
$k_n=(k_{n;0},\vec{0})^T$ and set $k_j=(k_{j,0},\vec{k_j})^T$. We also set
$\vec{k_j}\cdot \vec{k_j}=:t_j$ and have $k_j^2=k_{j;0}^2-t_j$, and finally define $\hat{k_j}=\vec{k_j}/\sqrt{t_j}$. Hence, 
\[
d^Dk_j=dk_{j,0}\;\sqrt{t_j}^{D-3} dt_j\;d\hat{k}_j,
\] with an angular measure 
\[
\int_{S^{D-2}} d\hat{k}_j\,1=\omega_{\frac{D}{2}}.
\]
Here,
\be\label{omegadhalf} 
\omega_{\frac{D}{2}}=\frac{2\pi^{\frac{D-1}{2}}}{\Gamma(\frac{D-1}{2})},\,\Gamma\left(\frac{1}{2}\right)\equiv\sqrt{\pi}.
\ee
We then have as integrations
\[
\int_{\mathbb{M}^D}d^Dk_j f(k_j)=\int_{-\infty}^{\infty} dk_{j;0}\int_0^\infty \sqrt{t_j}^{D-3}dt_j\int_{S^{D-2}} d\hat{k}_j f(k_{j,0},t_j,\hat{k_j}).
\]

\section{Minimal subtraction}
For the reader which likes to compare with dimensional regularization and the use of minimally sutraction as renormalization we have kept $D$ complex in most formulae 
and note that in such a situation the coproduct for $b_n$ reads
\be\label{HopfBanana}
\Delta_H(b_n)=b_n\otimes \One+\One\otimes b_n+\sum_{x,|x|\lneq n} x\otimes t_{n-|x|}.
\ee
Here the sum is over all monomials $x$ of banana graphs $b_j$ on less than $n$ edges.
For example
\[
\Delta(b_5)=b_5\otimes\One+\One \otimes b_5+\left(5\atop 2\right) b_2\otimes t_3
+\left(5\atop 3\right) b_3\otimes t_2+\left(5\atop 4\right)b_4\otimes t_1
+\left(5\atop 2\right)\left(3\atop 2\right)b_2b_2\otimes t_1.
\]
In Feynman graphs this is Fig.(\ref{hopf_tad}).
\begin{figure}[H]
\includegraphics[width=14cm]{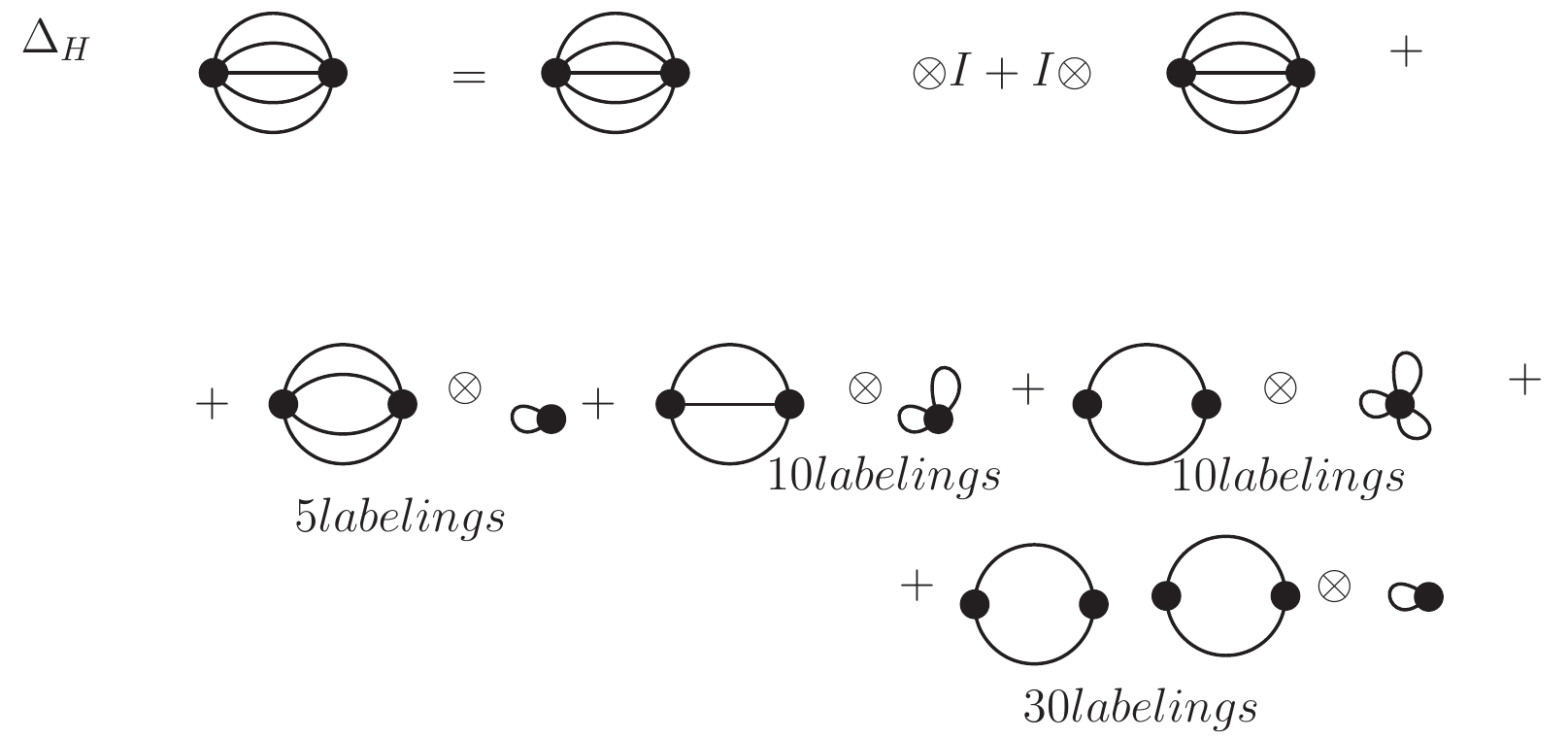}.
\caption{The Hopf algebra disentangling the five-banana $b_5$. On the right we also get roses with $n$ petals, or tadpoles in a physicists parlance. There are 
$5=\left(5\atop 4\right)$ labelings for the $b_4$ banana in the first term in the second row, and $10=\left(5\atop 3\right)=\left(5\atop 2\right)$ for the next two tensorproducts. The final term in the third row has $30=\left(5\atop 2\right)
\left(3\atop 2\right)$ labelings, as there are $\left(5\atop 2\right)$ possibilities to label the edges of the first $b_2$ banana, and then $\left(3\atop 2\right)$ to label the second one.}
\label{hopf_tad}
\end{figure}
Explicitely $\Phi_{MS}^D(b_3)$ reads for example
\[
\Phi_{MS}^D(b_3)=-\langle \Phi^D(b_3)\rangle+\sum_{cycl}\langle\langle
\Phi^D(b_2(e_i,e_j))\rangle \Phi^D(t_1(e_k))\rangle+\Phi^D(b_3)-
\sum_{cycl}\langle
\Phi^D(b_2(e_i,e_j))\rangle \Phi^D(t_1(e_k)).
\]
Here, $\Phi^D$ are unrenormalized Feynman rules in $D$ dimensions
which evalute into a Laurent series in $D-2n$, $n$ a suitable integer, $\langle\ldots\rangle$ is the projection onto the pole part and the sum is over the three cyclic permutations of $i,j,k$. 

This MS-renormalization results $\Phi_{MS}^D(b_n)$ can be  related to $\Phi_R^D(b_n)$ if so desired.  See also the discussion with regards to MS and tadpoles in \cite{KarenDirk}.

\section{Tensor structure}\label{apptens}
\subsection{Tensor integrals}\label{tensorint}
We are not iterested in $\Phi_R^D(b_n)$ alone. To satisfy the needs of computational practice we should also  raise the powers of quadrics 
by taking derivatives $\partial_{m_j^2}^k$ with respect to mass-squares $m_j^2$
and we should allow scalar products $k_i\cdot k_j$ in the numerator.

For such a  generalization to arbitrary powers of propagators and numerator structures
we use the notation 
\[
\Phi_R^D(b_n)_{\nu}(s,\{m_i^2\}),
\]
where $\nu$ is a $\left(\frac{n(n+1)}{2}-1\right)$-dimensional row vector
with integer entries  (see Sec.(2.5.1.)) in \cite{Zay}.

\begin{itemize}
\item The first $n$ entries $\nu_i$, $1\leq i\leq n$ give the powers of the $n$ edge propagators $\frac{1}{Q_e}$,
\item  the $(n-2)$ entries $\nu_i$, $(n+1)\leq i\leq (2n-2)$ are reserved for
powers of $k_i\cdot k_n$ ($1\leq i\leq (n-2)$), 
\item the   $(n-2)$ entries $\nu_i$, $(2n-1)\leq i\leq (3n-4)$ are reserved for
powers of $k_2^2,\ldots, k_{n-1}^2$,
\item and the remaining $(n-2)(n-3)/2$
entries are reserved for powers $\nu_{jl}$ of $k_j\cdot k_l$, $|j-l|\gneq 1$, $1\leq j,l\leq (n-1)$ and $3n-3\leq i\leq \left(\frac{n(n+1)}{2}-1\right)$.   
\end{itemize} 
This is all what is needed as $k_1^2=Q_1+m_1^2$ and
$2k_i\cdot k_{i-1}=k_i^2+k_{i-1}^2-Q_i-m_i^2$, $n\geq i\geq 2$.

For example
\beas
\Phi_R^D(b_4)_{(\nu_1,\ldots,\nu_{13})}(s,m_1^2,\ldots,m_4^2) & = &
\int_{\mathbb{M}_4}\omega^D_{(3)}\times\\
 & \times & \prod_{j=0}^{3}\frac{(k_1\cdot k_4)^{\nu_5}(k_2\cdot k_4)^{\nu_6}(k_2^2)^{\nu_7}(k_3^2)^{\nu_8}(k_1\cdot k_3)^{\nu_{13}}}{\left((k_{j+1}-k_j)^2-m_{j+1}^2\right)^{\nu_{j+1}}}.
\eeas  
For the imaginary part we have correspondingly
\beas
\Im\left(\Phi_R^D(b_4)_{(\nu_1,\ldots,\nu_{13})}\right)(s,m_1^2,\ldots,m_4^2) & = &
\int_{\mathbb{M}_4}\omega^D_{(3)}\times\\
 & \times & \prod_{j=0}^{3}\partial_{m_{j+1}^2}^{\nu_{j+1}}
 \left(
\left( \prod_{l=0}^{3}\delta_+(k_{l+1}-k_l)^2-m_{l+1}^2)\right)\times \right.\\
  & \times & \left. (k_1\cdot k_4)^{\nu_5}(k_2\cdot k_4)^{\nu_6}(k_2^2)^{\nu_7}(k_3^2)^{\nu_8}(k_1\cdot k_3)^{\nu_{13}}\right. \Biggr).
\eeas  

We  discuss differential equations for $\Phi_R^D(b_n)_\nu$, as well as partial integration identities and the reduction to master integrals starting from our representation for $\Phi_R^D(b_n)_\nu$ in Secs.(\ref{diffeq},\ref{masterint}). 
\subsection{Dispersion for $\Phi_R^D(b_n)_\nu$}
For banana graphs $b_n$ on two vertices dispersion for tensor integrals is rather simple:
\be 
\Phi_R^D(b_n)_\nu(s,s_0,\{m_j^2\})=\frac{(s-s_0)^{|[n,\nu]|}}{\pi}
\int_{(m_1+\cdots m_n)^2}^\infty
 \frac{V_{[n,\nu]}^D}{(x-s)(x-s)^{|[n,\nu]|}} dx,
\ee
where $|[n,\nu]|-1$ is the superficial degree of divergence of $\Phi_R^D(b_n)_\nu$
according to $\nu$:
\be\label{degdiv}
|[n,\nu]|=\left(\frac{D}{2}-1\right)(n-1)+\sum_{j=1}^n \nu_j
+\left\lceil \sum_{j=n+1}^{2n-2}\frac{\nu_j}{2}\right\rceil +\sum_{j=2n-1}^{3n-4}\nu_j
+\sum_{jl} \nu_{jl}.
\ee
This is based on
\[
\Im\left(\Phi_R^D(b_n)_\nu\right)(s,s_0,\{m_j^2\})=
 \Theta(s-(m_1+\cdots m_n)^2)V_{[n,\nu]}^D.
\]
For $V_{[n,\nu]}^D$ see Eqs.(\ref{tensors}-\ref{tensorsthree}) below.
\section{pseudo-thresholds}\label{apppseudo}
Let us remind ourselves of a parametric analysis of the second Symanzik polynomial
(with masses) $\Phi$ for the banana graphs $b_b$:
\be\label{varphibn}
\varphi(b_n)(s)=\left(\prod_{j=1}^n A_j\right)\left(s-\left(\sum_{j=1}^n m_j^2A_j\right)\left(\sum_{j=1}^n\frac{1}{A_j}\right)\right).
\ee

The equation
\[
\varphi(b_n)(m_{\mathbf{normal}}^n)=0,
\]
has a solution in the simplex $A_i>0$ for positive $A_i$ given by $A_im_i=A_jm_j$.

For $m$ any pseudo-mass, the solution of 
$\varphi(b_n)(m)=0$ requires at least one $A_i$ to be negative and it hence gives no
monodromy on the physical sheet.

Still the variations associated to pseudo-masses and thresholds are needed for a full analysis of $\Phi_R^D(b_n)$ to find their Hodge structure.

So let $\sigma_n$ be a sequence of the form 
\[
\sigma^n:=(\pm m_1\pm m_2\pm\cdots\pm m_n),
\]
with a sign chosen for each entry $m_i$. Let $p(i)\in\{\pm 1\}$ be the sign of the $i$-entry. A global sign change leaves the pseudo-thresholds invariant ($|a-b|=|b-a|$) so we have $2^{n-1}$ choices and adopt to the convention $p(1)=+1$. 
 
For a flag
\[
(b_2\subset b_3\subset\cdots\subset b_n),
\]
this determines subsequences $\sigma^2\subset \sigma^3\subset\cdots  \sigma^n$
in an obvious manner.

Define
\be\label{pseudoup}
\mathrm{up}_n^{j,\sigma}:=\frac{s_n^j+m_{n-j}^2-\Bigl(\overbrace{\sum_{||,i=1}^{n-j-1}p(i)m_i}^{m_{\sigma^{n-j-1}}}\Bigr)^2}{2\sqrt{s_n^{j}}},
\ee
which also defines the pseudo-mass $m_{\sigma^{n-j-1}}$:
\[
m_{\sigma^{n-j-1}}=\sum_{||,i=1}^{n-j-1}p(i)m_i=\underbrace{|\cdots||}_{(n-1)\,{\text{bars}}}m_1+p(2)m_2|+p(3)m_3|+\cdots|+p(n-j-1)m_{n-j-1}|.
\]
Define
\[
\Theta_{n,+}=\Theta(s-(m_n+m_{\sigma^{n-1}})^2),\,
\Theta_{n,-}=\Theta((m_n-m_{\sigma^{n-1}})^2-s).
\]
Now set for $p(n-1)=+1$:
\beas
 & &\mathbf{Var}(b_n^\sigma)
  = \Theta_{n,p(n))}\times\\
 & & \underbrace{\omega_{\frac{D}{2}}\int_{m_n}^{\mathrm{up}_n^{0,\sigma}}V_{\sigma^{n-1},n-1}^D(s_n^0-2\sqrt{s_n^0}y_{n-1}+m_n^2,m_1^2,\ldots,m_{n-1}^2)
\sqrt{y_{n-1}^2-m_n^2}^{D-3}dy_{n-1}}_{V^D_{\sigma^n,n},\,p(n-1)=+1}. 
\eeas
and for $p(n-1)=-1$:
\beas
 & &\mathbf{Var}(b_n^\sigma)
  = \Theta_{n,p(n))}\times\\
 & & \underbrace{\omega_{\frac{D}{2}}\int_{\mathrm{up}_n^{0,\sigma}}^{\infty}V_{\sigma^{n-1},n-1}^D(s_n^0-2\sqrt{s_n^0}y_{n-1}+m_n^2,m_1^2,\ldots,m_{n-1}^2)
\sqrt{y_{n-1}^2-m_n^2}^{D-3}dy_{n-1}}_{V^D_{\sigma^n,n},\,p(n-1)=-1}. 
\eeas
Apart from the variation for the normal threshold (with $p(i)=+1$ for all $1\leq i\leq n$) which gives $\mathbf{Var}(b_n^{(+m_1,+m_2,\ldots,+m_n)})=\Im\left(\Phi_R^D(b_n)\right)$,
we get $2^{n-1}-1$ further variations corresponding to pseudo-masses and their pseudo-thresholds. They will be discussed elsewhere.

\section{$b_3$ parametrically}\label{bthreeparamet}
Let us recapitulate $b_3$ in the parametric representation.
We list basic considerations. A detailed analysis in the view of \cite{Marko}
and \cite{MarkoDirk} is left to future work.
\subsection{The parametric integral}
Let $\mathbf{Q}_{b_3}$ be the one-dimensional real vetorspace spanned by $s=k_3^2$, the square of the Minkowski four-momenta $k_3,-k_3$ assigned to the two vertices of $b_3$.
Let $\mathbb{P}_{b_3}=\mathbb{P}^2(\mathbb{R}_+)$ be a projective space given by the ratios of the nonnegative side-lengths of the internal edges of $\Theta$.

The parametric integrand function (we consider masses as implicit parameters)
\[
F_{b_3}:\mathbf{Q}_{b_3}\times\mathbf{Q}_{b_3}\times\mathbb{P}_{b_3}\to \mathbb{C}
\]
is (see for example Sec.(5.2.1.) in \cite{BrownKr})
\bea\label{bthreepar}
F_{b_3}(s,s_0;p) & := & (s-s_0)A_1A_2A_3\frac{\ln\left(\frac{\Phi_\Theta(s;p)}{\Phi_\Theta(s_0;p)}\right)}{\psi_\Theta^3}\\
 & + &
(s_0A_1A_2A_3-(m_1^2A_1+m_2^2A_2+m_3^2A_3)\psi_\Theta) \times\nonumber\\
 & \times & \frac{\ln\left(\frac{\Phi_\Theta(s;p)}{\Phi_\Theta(s_0;p)}\right)-
(s-s_0)\left(\partial_{s}\ln\left(\frac{\Phi_\Theta(s;p)}{\Phi_\Theta(s_0;p)}\right)
\right)_{s=s_0}}{\psi_\Theta^3}.\nonumber
\eea
Here,
\[
\Phi_{b_3}: \mathbf{Q}_\Theta\times \mathbb{P}_\Theta\to \mathbb{C}
\]
is
\[
\Phi_{b_3}(r;p)=rA_1A_2A_3-(m_1^2A_1+m_2^2A_2+m_3^2A_3)\psi_{b_3},
\]
\[
\psi_{b_3}=A_1A_2+A_2A_3+A_3A_1.
\]
Note $F_{b_3}(s,p)$ and $\partial_{s} F_{b_3}(s,p)$ both vanish at $s=s_0$ for all $p$, so these are on-shell renormalization conditions.

The parametric form is the integrand
\[
\mathrm{Int}_{b_3}(s,s_0;p):=F_\Theta(s,s_0,p) \Omega_{b_3},
\]
\[
\Omega_{b_3}=+A_1\,dA_2\wedge dA_3-A_2\,dA_1\wedge dA_3+A_3\,dA_1\wedge dA_2.
\]

We then have the renormalized value\footnote{Divergent subgraphs exist but do not need renormalization as the cographs are tadpoles which can be set to zero in kinematic renormalization. Accordingly $F_\Theta$ vanishes when any two of its three
edge variables $A_i$ vanish.}
\be\label{bthreep}
\Phi_R^D(b_3)(s,s_0)=\int_{\mathbb{P}^2(\mathbb{R}_+)}\mathrm{Int}_\Theta(s,s_0;p),
\ee
from integrating out $p$ which is the parametric equivalent of Eqs.(\ref{bthreeo},\ref{bthreed}).
\subsection{Sectors and fibrations}
To study fibrations in our integrand we start from the fact that there  are six orderings of the edge lengths for the three edge variables $A_i$.

Consider for example the sectors $1>3>2$ and $3>1>2$ of Fig.(\ref{theta}) so that edge $e_2$ has the smallest length. For the choice $1>3>2$
rescale\footnote{$\Omega_{b_3}\to A_1^3 da_2\wedge da_3$ under that rescaling.}
\[
A_2=a_2A_1,\,A_3=a_3A_1,
\]
and in that sector $1>3>2$
we have
\[
\int_{\mathbb{P}^2(\mathbb{R}_+)\cap (1>3>2)}F_{b_3} \Omega_{b_3}
=\int_0^\infty\left(\int_0^{a_3}F_{b_3}(1,a_2,a_3)da_2\right)da_3.
\]
A further change $a_2=a_3b_2$ leads to a sector decomposition (in the sense of physicists)
\[
\int_0^\infty\left(\int_0^{1}a_3F_{b_3}(1,b_2a_3,a_3)db_2\right)da_3
=
\int_0^{1}\underbrace{\left(\int_0^\infty a_3F_{b_3}(1,b_2a_3,a_3)da_3\right)}_ {\mathrm{Fib}(b_2)}db_2.
\]
For any chosen $0<b_2<1$, $a_3F_{b_3}(1,b_2a_3,a_3)$ gives points on the corresponding  chosen fiber and $\mathrm{Fib}(b_2)$ is the integral along that fiber. Integrating $b_2$ integrates all fiber integrals  $\mathrm{Fib}(b_2)$ to the two sector integrals on both sides of the spine. 

In fact for $0<a_3<m_1/m_3$ we are on the left of the spine and for
$m_1/m_3<a_3<\infty$ on the right. 

Let us look at $\Phi_{b_3}$ under the rescalings.
\beas
\Phi_{b_3}(A_1,A_2,A_3) & = & sA_1A_2A_3-(m_1^2A_1+m_2^2A_2+m_3^2A_3)(A_1A_2+A_2A_3+A_3A_1))\\
 & \to &  sa_2a_3-(m_1^2+m_2^2a_2+m_3^2a_3)(a_2+a_2a_3+a_3))\\
 & \to & sb_2a_3^2-(m_1^2+m_2^2b_2a_3+m_3^2a_3)(b_2a_3+b_2a_3^2+a_3))\\
  & = & a_3(sb_2a_3-(m_1^2+m_2^2b_2a_3+m_3^2a_3)(b_2+b_2a_3+1))=:\tilde{\Phi}_{b_3}(s,b_2,a_3).
\eeas
For $\psi_{b_3}$ we find
\beas
 & & (A_1A_2+A_2A_3+A_3A_1)\\
  & \to & (a_2+a_2a_3+a_3)\\
   & \to & a_3(b_2+b_2a_3+1).
\eeas
We thus find in the region where $e_2$ is the smallest edge the integrand function 
$\mathrm{Int}_{{b_3},2}(b_2,a_3)$
\beas
\mathrm{Int}_{{b_3},2}(b_2,a_3) & := & a_3F_{b_3}(1,b_2a_3,a_3)=(s-s_0)b_2a_3\times\\
 & \times & \frac{\ln
\overbrace{\left(
\frac{sa_3b_2-(m_1^2+m_2^2b_2a_3+m_3^2a_3)(1+b_2(1+a_3))}{s_0a_3b_2-(m_1^2+m_2^2b_2a_3+m_3^2a_3)(1+b_2(1+a_3))}\right)}^{\frac{\tilde{\Phi}_{b_3}(s;b_2,a_3)}{\tilde{\Phi}_{b_3}(s_0;b_2,a_3)}}}{(b_2(1+a_3)+1)^3}\\
 & + & 
(s_0b_2a_3-(m_1^2+m_2^2b_2a_3+m_3^2a_3)(b_2(1+a_3)+1)) \times\\
 & \times & \frac{
 \ln\left(\frac{\tilde{\Phi}_{b_3}(s;b_2,a_3)}{\tilde{\Phi}_{b_3}(s_0;b_2,a_3)}\right)-
(s-s_0)
\left(\partial_{s}\ln\left(\frac{\tilde{\Phi}_{b_3}(s;b_2,a_3)}{\tilde{\Phi}_{b_3}(s_0;b_2,a_3)}
\right)\right)_{s=s_0}}{(b_2(1+a_3)+1)^3}.
\eeas
Note that $\mathrm{Int}_{{b_3},2}(0,a_3)=0$ as it must be as petals evaluate to zero under renormalized Feynman rules in on-shell renormalization conditions.

Finally
\[
{\mathrm{Fib}(b_2)}=\int_0^\infty \mathrm{Int}_{{b_3},2}(b_2,a_3)da_3. 
\]
A point along the $(1=3)$-line of the spine is given by $(1,b_2,1)\in \mathbb{P}_{b_3}$, for all $0<b_2<1$. 
\begin{rem} Upon rescaling in each of the sectors in the three cubes of Fig.(\ref{theta}) accordingly and summing over sectors we get a symmetric representation equivalent to averaging over the three possible ways of expressing Eq.(\ref{bthreeconcrete}) using any of $s_3^1(y_2,m_i^2)$ and
similar to \cite{DavDel}.
\end{rem}\hfill$|$
\end{appendix}


\begin{thebibliography}{99}
\enlargethispage{2\baselineskip}
\bibitem{Ralph} R.M.\ Kaufmann, S.\ Khlebnikov, B.\ Wehefritz-Kaufmann,
{\em Singularities, swallowtails and Dirac points. An analysis for
families of Hamiltonians and applications to wire networks,
especially the Gyroid}, Annals of Physics 327 (2012) 2865–2884.
\bibitem{Veltman} M.\ Veltman, {\em Unitarity and Causality in a renormalizable field theory with unstable particles}, Physica {\bf 29} 186 (1963).
\bibitem{BroedeletalEll}J.\ Br\"odel, 
C.\ Duhr, F.\ Dulat, B.\ Penante,  L.\ Tancredi, {\em Elliptic Feynman integrals and pure functions}, J. High Energ.Phys.{\bf 2019} 23 (2019),
arXiv:1809.10698 [hep-th].
\bibitem{Broedelequalmass} J.\ Broedel, C.\ Duhr, F.\ Dulat, R.\ Marzucca, B.\ Penante, L.\ Tancredi, {\em An analytic solution for the equal-mass banana graph},
JHEP{\bf 09} 112 (2019).  
\bibitem{Remetal}  M.\ Caffo, H.\ Czy\.z, S.\ Laporta, E.\ Remiddi, {\em The Master Differential Equations for the 2-loop Sunrise
Selfmass Amplitudes.}, Nuovo Cimento {\bf A} 111(4),
 hep-th/9805118.
\bibitem{RemSchouten} E.\ Remiddi, L.\ Tancredi, 
{\em Schouten identities for Feynman graph amplitudes;
the Master Integrals for the two-loop massive sunrise graph}, 
Nuclear Physics {\bf B880}, 343 (2014),
arXiv:1311.3342 [hep-th].
\bibitem{Weinzierl}
 L.\ Adams, C.\ Bogner, S.\ Weinzierl,
  {\em The sunrise integral and elliptic polylogarithms,}
  PoS LL {\bf 2016} (2016) 033
  doi:10.22323/1.260.0033
  [arXiv:1606.09457 [hep-ph]].
\bibitem{BlochVanhove}
S.\ Bloch, P.\ Vanhove,
{\em The elliptic dilogarithm for the sunset graph,}
J.\ Number Theor.\ \textbf{148} (2015), 328-364
doi:10.1016/j.jnt.2014.09.032
[arXiv:1309.5865 [hep-th]].
\bibitem{Vanhove}
S.\ Bloch, M.\ Kerr, P.\ Vanhove,  {\em Local mirror symmetry and the sunset Feynman integral,}
  Adv.\ Theor.\ Math.\ Phys.\  {\bf 21} (2017) 1373
  doi:10.4310/ATMP.2017.v21.n6.a1
  [arXiv:1601.08181 [hep-th]].   
\bibitem{Bloch} S.~Bloch, M.~Kerr and P.~Vanhove,
{\em A Feynman integral via higher normal functions,}
Compos. Math. \textbf{151} (2015) no.12, 2329-2375
doi:10.1112/S0010437X15007472
[arXiv:1406.2664 [hep-th]].  
\bibitem{DavDel} A.\ Davydychev, R.\ Delbourgo, {\em Explicitly symmetrical treatment of three-body phase space}, J.Phys.A37:4871-4886,2004, hep-th/0311075.
\bibitem{Zay} R.\ Zayadeh, {\em Picard–Fuchs Equations of Dimensionally Regulated
Feynman Integrals}, thesis Mainz University,
https://openscience.ub.uni-mainz.de/bitstream/20.500.12030/3696/1/3663.pdf
\bibitem{allBanana} K.\ B\"onisch,  F.\ Fischbach,  A.\ Klemm, C.\ Nega,
R.\ Safari, {\em Analytic Structure of all Loop Banana Amplitudes},
J.High Energ.Phys.{\bf 2021} 66 (2021), 
arXiv:2008.10574 [hep-th].
\bibitem{BroadKloster} D.\ Broadhurst, {\em Feynman integrals, L-series and
Kloosterman moments},
Communications in
Number Theory and Physics {\bf 10}, No.3, 527–569 (2016).
\bibitem{Kerwas} B.P.\ Kersevan, E.\ Richter--Was,  {\em Improved Phase Space Treatment of Massive Multi-Particle Final States}, 
	Eur.Phys.J.{\bf C39} (2005) 439-450,  hep-ph/0405248.
\bibitem{Block} M.M.\ Block, {\em Phase-Space Integrals for Multiparticle Systems},
Phys.\ Rev.\ {\bf 101} 796 (1956).
\bibitem{Prem} P.P.\ Srivastava, G.\ Sudarshan, {\em Multiple Production of Pions in Nuclear Collisions}, Phys.\ Rev.\ {\bf 110} 765 (1958).
\bibitem{Brown} F.\ Brown, {\em Invariant Differential Forms on Complexes of Graphs and Feynman Integrals
}, 	SIGMA {\bf 17} 103 (2021).
\bibitem{BEK} S.\ Bloch, H.\ Esnault, D.\ Kreimer, {\em On Motives Associated to Graph Polynomials}, : Commun.Math.Phys.{\bf 267} (2006) 181-225.
\bibitem{Broedel} J.\ Broedel, C.\ Duhr, N.\ Matthes, {\em Meromorphic modular forms and the three-loop equal-mass banana integral}, arXiv:2109.15251 .
\bibitem{Coleman-N} S.\ Coleman, R.\ Norton, {\em Singularities in the physical region}, Nuovo Cimento {\bf 38}, 438 (1965).
\bibitem{DirkEll} D.\ Kreimer, {\em Multi-valued Feynman Graphs and Scattering Theory},  Elliptic Integrals, Elliptic Functions and Modular Forms in Quantum Field Theory, Texts \& Monographs in Symbolic Computation 2019, J.Bluemlein et.al., eds
\bibitem{BlKrCut} Spencer Bloch, Dirk Kreimer,
{\em Cutkosky Rules and Outer Space}, arXiv:1512.01705.
\bibitem{Tkachov} K.\ Chetyrkin, F.\ Tkachov, {\em Integration by parts: The algorithm to calculate $\beta$-functions in 4 loops},
Nuclear Physics {\bf B192}, 23 (1981). 
\bibitem{algorithms}
S.\ Laporta, {\em High-precision calculation of multi-loop Feynman integrals by difference equations}, Int.J.Mod.Phys.{\bf A 15}, 5087 (2000).
\bibitem{RemODE} E.\ Remiddi, {\em Differential Equations for Feynman Graph Amplitudes}, Nuovo Cim.{\bf A110} (1997) 1435-1452, hep-th/9711188.
\bibitem{KalmKniehl} M.\ Kalmykov, B.\ Kniehl, {\em Counting the number of master integrals for sunrise
diagrams via the Mellin–Barnes representation}, 
	JHEP {\bf 1707} (2017) 031,
	arXiv:1612.06637 [hep-th].
\bibitem{Panzer} T.~Bitoun, C.~Bogner, R.~P.~Klausen and E.~Panzer,
{\em Feynman integral relations from parametric annihilators,}
Lett.Math.Phys.\textbf{109} (2019) no.3, 497-564,
[arXiv:1712.09215 [hep-th]].
\bibitem{KarenDirk} D.\ Kreimer, K.\ Yeats, {\em Algebraic Interplay between Renormalization and Monodromy,}
[arXiv:2105.05948 [math-ph]].
\bibitem{DirkCoact}
D.~Kreimer,
{\em Outer Space as a Combinatorial Backbone for Cutkosky Rules and Coactions,}
doi:10.1007/978-3-030-80219-6${}_{12}$,
arXiv:2010.11781 [hep-th].
\bibitem{CullVogt}
M.\ Culler, K.\ Vogtmann {\em Moduli of graphs and automorphisms
of free groups} (1986) Invent.\ Math.\ {\bf  84}(1):91–119.
\bibitem{intersect}
P.\ Mastrolia and S.\ Mizera,
{\em Feynman Integrals and Intersection Theory,}
JHEP \textbf{02} (2019), 139
doi:10.1007/JHEP02(2019)139
[arXiv:1810.03818 [hep-th]].
\bibitem{VanhRev} P.\ Vanhove, {\em Feynman integrals, toric geometry and mirror
symmetry}, 
In: J.\ Bl\"umlein, C.\ Schneider, P.\ Paule (eds) {\em Elliptic Integrals, Elliptic Functions and Modular Forms in Quantum Field Theory}. Texts \& Monographs in Symbolic Computation, Springer.
\bibitem{Carlsen} B.C.\ Carlson, {\em Special Functions of Applied Mathematics},
 AP (1977).
\bibitem{DirkCarlsen} L.\ Brucher, J.\ Franzkowski and D.\ Kreimer,
{\em Loop integrals, R functions and their analytic continuation,}
Mod.Phys.Lett.A \textbf{9} (1994), 2335-2346,
[arXiv:hep-th/9307055 [hep-th]].
\bibitem{Cutkosky}
R.~E.~Cutkosky,
{\em Singularities and discontinuities of Feynman amplitudes,},
J.Math.Phys.\textbf{1} (1960), 429-433,
doi:10.1063/1.1703676
\bibitem{Marko} M.\ Berghoff, {\em Feynman amplitudes on moduli spaces of graphs},
	Ann.Inst.Poincar\'e {\bf D7}, Iss.2, 203  (2020),  arXiv:1709.00545.
\bibitem{MarkoDirk}
M.~Berghoff and D.~Kreimer,
{\em Graph complexes and Feynman rules,}
[arXiv:2008.09540 [hep-th]].
\bibitem{BrownKr}
F.~Brown and D.~Kreimer,
{\em Angles, Scales and Parametric Renormalization,}
Lett.Math.Phys.\textbf{103} (2013), 933-1007,
doi:10.1007/s11005-013-0625-6,
[arXiv:1112.1180 [hep-th]].
\end{thebibliography}
\end{document}